\documentclass[11pt]{article}

\usepackage{fullpage,times}
\usepackage{amsmath}
\usepackage{amsthm}
\usepackage{amssymb}
\usepackage{amsfonts}
\usepackage{graphics}
\usepackage{color}
\usepackage{graphicx}

\usepackage{latexsym}
\usepackage{euscript}
\usepackage{amstext}
\usepackage{color}
\usepackage{dsfont}
\usepackage{xspace}
\usepackage{url,hyperref}

\usepackage{algorithm,algorithmic}

\newcommand\ignore[1]{}

\newcommand{\squishlist}{
 \begin{list}{$\bullet$}
  { \setlength{\itemsep}{0pt}
     \setlength{\parsep}{3pt}
     \setlength{\topsep}{3pt}
     \setlength{\partopsep}{0pt}
     \setlength{\leftmargin}{1.5em}
     \setlength{\labelwidth}{1em}
     \setlength{\labelsep}{0.5em} } }

\newcommand{\squishend}{
  \end{list}  }

\newtheorem{lemma}{Lemma}[section]
\newtheorem{theorem}[lemma]{Theorem}

\newtheorem{proposition}[lemma]{Proposition}
\newtheorem{claim}[lemma]{Claim}
\newtheorem{remark}[lemma]{Remark}

\newcommand{\etal}{et al.\ }
\newcommand{\eps}{\epsilon}

\newcommand{\cR}{\mathcal{R}}
\newcommand{\cA}{\mathcal{A}}
\newcommand{\cG}{\mathcal{G}}
\newcommand{\cQ}{\mathcal{Q}}

\newcommand{\cP}{\mathcal{P}}
\newcommand{\cT}{\mathcal{T}}
\newcommand{\cB}{\mathcal{B}}
\newcommand{\cI}{\mathcal{I}}

\newcommand{\Ex}{\mathbb{E}}
\newcommand{\Var}{\texttt{Var}}

\newcommand{\opt}{\textsf{OPT}}

\newcommand{\stemp}{\sigma_{temp}}
\newcommand{\sfinal}{\sigma_{final}}

\newcommand{\fifo}{\textsf{FIFO}}

\newcommand{\lpmflow}{\mathsf{LP}_\mathsf{MaxFlow}}

\newcommand{\prev}{\texttt{Prev}}

\newcommand{\qtemp}{Q_{temp}}

\newcommand{\Ftemp}{F^{\sigma_{temp}}}
\newcommand{\Ffinal}{F^{\sigma_{final}}}
\newcommand{\overflow}{OF}

\newenvironment{proofof}[1]{\smallskip\noindent{\bf Proof of #1}}%
        {\hspace*{\fill}$\Box$\par}

\newcommand{\initOneLiners}{%
    \setlength{\itemsep}{0pt}
    \setlength{\parsep }{0pt}
    \setlength{\topsep }{0pt}
}
\newenvironment{OneLiners}[1][\ensuremath{\bullet}]
    {\begin{list}
        {#1}
        {\initOneLiners}}
    {\end{list}}

\begin{document}
\title{Optimizing Maximum Flow Time and Maximum Throughput \\ in Broadcast Scheduling}

\author{
Sungjin Im\thanks{Department of Computer Science, Duke University, Durham NC 27708-0129. {\tt
sungjin@cs.duke.edu}. This work was partially done while the author was
visiting University of Warwick. Partially supported by NSF grant CCF-1008065.}
\and Maxim Sviridenko\thanks{University of Warwick. 
{\tt sviri@dcs.warwick.ac.uk}, Work supported by EPSRC grants EP/J021814/1, EP/D063191/1,  FP7 Marie Curie Career Integration Grant and Royal Society Wolfson Research Merit Award. }
}

\date{}
\maketitle

\begin{abstract}
    We consider the pull-based broadcast scheduling model. In this model, there are $n$ unit-sized pages of information available at the server. Requests arrive over time at the server asking for a specific page. When the server transmits a page, all outstanding requests for the page are simultaneously satisfied, and this is what distinguishes broadcast scheduling from the standard scheduling setting where each job must be processed separately by the server. Broadcast scheduling has received a considerable amount of attention due to the algorithmic challenges that it gives in addition to its applications in multicast systems and wireless and LAN networks. In this paper, we give the following new approximation results for two popular objectives:
    \begin{OneLiners}
        \item For the objective of minimizing the maximum flow time, we give the \emph{first} PTAS. Previously, it was known that the algorithm First-In-First-Out ($\fifo$) is a 2-approximation, and it is tight \cite{ChangEGK11, ChekuriIM12}. It has been suggested as an open problem to obtain a better approximation \cite{ChangEGK11, Dagstuhlscheduling10, Im12thesis, Moseley12thesis}.
        \item For the objective of maximizing the throughput, we give a $0.7759$-approximation which improves upon the previous best known $0.75$-approximation \cite{GandhiKPS06}.
    \end{OneLiners}

    Our improved results are enabled by our novel rounding schemes and linear programming which can effectively reduce congestion in schedule which is often the main bottleneck in designing scheduling algorithms based on linear programming. We believe that our algorithmic ideas and techniques could be of  potential use for other scheduling problems.
\end{abstract}

\section{Introduction}

We revisit the pull-based broadcast scheduling model. In this model, there is a single server that stores $n$ pages of useful information. Each client sends a request $\rho$ asking for a specific page $p$. When the server broadcasts a page $p$, all outstanding/unsatisfied requests for the same page $p$ are satisfied simultaneously. This is the main difference from the standard scheduling setting where each request needs to be processed separately by the server. This model is called pull-based since clients initiate the requests, while in the push-based model the server transmits pages according to the given frequency of pages requested.

Broadcast scheduling has applications in multicast systems, LAN and wireless systems \cite{Wong88,AcharyaFZ95,AksoyF98}. We note that data broadcast scheduling is used in commercial systems \cite{Intercast, DirectPC, Airmedia}, and it helps increase the system's bandwidth by serving multiple requests simultaneously.
Also it can be viewed as a special case of batch scheduling that has been extensively studied in the stochastic and queueing theory literature \cite{DebS73,Deb84,Weiss79,WeissP81}.
Broadcast scheduling has received a substantial amount of attention from the algorithms community, and has been studied both for the pull-based and push-models, also both in the online and offline settings \cite{BarnoyBNS98,AksoyF98,AcharyaFZ95,BartalM00}. This is because in addition to the aforementioned applications it gives algorithmic challenges concerning how to group requests for the same page over time to satisfy more requests with less transmissions while optimizing/satisfying certain objectives.

In this paper, we consider two objectives of minimizing the maximum flow time and maximizing the total throughput (profit). We first discuss the first objective. Each request $\rho$ is released at time $r_\rho$ asking for a specific page $p_\rho$. We assume that all pages are unit-sized, and requests arrive only at  integer times. This unit-sized page assumption has been adopted in the most previous literature. This assumption is justified when all pages have similar size,  and still keeps the main difficulty of the problem. Consider any feasible schedule $\sigma$ where at most one page is transmitted at each integer time. The completion time $C^\sigma_\rho$ of request $\rho$ is defined as the first time greater than $r_\rho$ when page $p_\rho$ is transmitted. If no such transmission exists, $C^\sigma_\rho = \infty$.  Note that all requests have a flow time of at least one. The goal is to find a schedule $\sigma$ that minimizes $\max_\rho (C^\sigma_\rho - r_\rho)$. If the schedule $\sigma$ is clear from the context, it may be omitted.

This problem was first suggested in \cite{BartalM00}, which was the paper with \cite{KalyanasundaramPV00} that initiated the study of pull-based broadcast scheduling in the worst case analysis model. In fact, \cite{BartalM00} claimed that the online algorithm First-In-First-Out ($\fifo$) is 2-competitive for this problem. However, it was fairly later that the formal proof was found \cite{ChangEGK11, ChekuriIM12}. This problem was shown to be NP-hard \cite{ChangEGK11}. Although the simple algorithm $\fifo$ achieves $2$-competitiveness, it has been the best known approximation guarantee even for the offline setting. Hence a natural open question was if one can obtain a better approximation in the offline setting \cite{ChangEGK11, Im12thesis, Moseley12thesis}. Furthermore, this problem was mentioned in the Dagstuhl seminar on scheduling in 2010 as an open problem with an interesting connection to the so called IRS Tax Scheduling problem \cite{Dagstuhlscheduling10}.

In the other problem of maximizing the total throughput (profit), each request $\rho$ is also associated with deadline $d_\rho$ and profit $w_\rho$. If page $p$ is transmitted during $[r_\rho+1, d_\rho]$, the request yields profit $w_\rho$. This objective is also NP-hard to optimize \cite{ChangEGK11}. There are several constant factor approximations known. The simple greedy (online) algorithm that transmits page $p$ that satisfies the requests of the maximum total profit is known to be $2$-competitive \cite{KimC04}.  Other approaches are based on linear programming and rounding. As a high-level overview, the LP gives a fractional schedule $\{x_{p,t}\}$ over all pages $p$ and time steps $t$ such that $\sum_p x_{p,t} = 1$ for all time steps $t$. Here $x_{p,t}$ is the (possibly fractional) amount of page $p$ that is transmitted at time $t$. The independent rounding of picking one page at each time $t$ according to $x_{p,t}$ gives a $(1 - 1/e)$-approximation \cite{GandhiKPS06, ChekuriGIKLMMR10}.  The best approximation currently known uses the elegant dependent rounding in \cite{GandhiKPS06} which gives a rounding scheme for a bipartite graph while keeping some hard constraints (here the relationship between pages and times is described as a bipartite graph). The current best approximation factor is 0.75 \cite{GandhiKPS06}.

\subsection{Our Contributions and Techniques}

One of our main results is the \emph{first} polynomial time approximation scheme (PTAS) for the maximum flow time objective.

\begin{theorem}
    \label{thm:fmax-main-det}
    There exists a PTAS for minimizing the maximum flow time in broadcast scheduling. More precisely, for any $0<\eps \leq 1$, there exists a $(1 + \eps)$-approximation algorithm with running  time $m^{O(1/ \eps^4)}$ where $m$ is the number of requests.
\end{theorem}

One of the key algorithmic ideas in obtaining this result lies in our novel group-based $\alpha$-point rounding. The $\alpha$-point rounding has been useful in rounding fractional solutions for   scheduling problems. For examples and pointers, see \cite{BansalCS08, GQSS, SkutellaW11}. As mentioned before, the LP relaxation will give the amount $x_{p,t}$ by which page $p$ needs to be transmitted at time $t$. In the standard $\alpha$-point rounding, for each page $p$, one random value $\alpha_p$ is picked uniformly at random from $[0,1]$, and page $p$ is transmitted at times $t$ such that $\sum_{t' =1}^t x_{p, t'} \geq \alpha_p +k > \sum_{t' =1}^{t-1} x_{p, t'}$ for an integer $k$. The resulting (possibly infeasible) schedule has nice properties such as preserving the flow time of each request in expectation. However, it could result in a large congestion during an interval $I$. Namely, too many transmissions may be made during $I$ compared to $|I|$, the maximum number of transmissions that can be made during $I$. This overflow could be as large as the standard deviation $\theta(\sqrt{n})$ for some interval; recall that $n$ is the number of pages. To make this schedule feasible, transmissions are delayed by the amount of overflow.  In fact, this is why \cite{BansalCS08} had to iteratively solve a sequence of relaxed linear programs to avoid this large delay for the average flow objective. However, the upper bound on the overflow shown in \cite{BansalCS08} is $O( \log^2 n / \log \log n)$, which is too large for our goal of designing a PTAS. Furthermore, the overflow during any interval could have a more serious effect on our objective of minimizing the maximum flow time, while for other objectives such as the total flow time, the increase of flow time of some requests may be charged to other requests.

Our key idea is to partition pages into a small number of groups and to let all pages in the same group $g$ to share a single random value $\alpha'_g \in [0,1]$. Roughly speaking, the requests for any two different pages in the same group have substantially different release times. If the difference is more than $L^*$, the maximum flow time of an optimal solution, it can be assumed that two different pages in the same group are never transmitted at the same time, and this is precisely why the pages in the same group can share the same random variable. Although $\alpha_p$ for pages $p$ in the same group $g$ are completely determined by a single random value $\alpha'_g$, the quantities $\alpha_p$ are not necessarily the same. To make the number of transmission made for pages in the same group $g$ as close to the total (fractional) amount of transmission made for the pages as possible, we transmit a page for each group $g$ at times $t$ when the cumulative quantity $\sum_{p \in g} \sum_{t' =1}^t x_{p, t'}$ first exceeds $\alpha'_g +k$ for some integer $k$. To our best knowledge, this seemingly simple idea has never been used before, and we believe that it is worth further investigation for the potential use for other problems. By applying concentration inequalities with this small number of random variables (at most $O(L^*)$) , we are able to show that the overflow is only $O(\eps L^*)$ for any interval if $L^*$ is considerably big. We derandomize this process using the  method of pessimistic estimators \cite{Raghavan88}. When $L^*$ is small, we design a dynamic programming which completes our PTAS.


\begin{remark}
    The reader may wonder if the algorithm $\fifo$ can be strengthened by an LP. We however show that a natural LP-guided $\fifo$ achieves only a $2$-approximation. See Section~\ref{sec:LP-guided-fifo}.
\end{remark}

The other main result of this paper is an improved approximation for the maximum throughput objective.

\begin{theorem}
    \label{thm:max-throughput}
    For any $\eps >0$, there exists a $(\frac{1}{2} +  \frac{3}{4e} - \eps)$-approximation for maximizing the throughput (total profit) in broadcast scheduling $(\frac{1}{2} + \frac{3}{ 4e} > 0.7759)$. Furthermore the running time of the algorithm is in polynomial in $(1 / \eps)^{O(1 / \eps)}$ and $m$. 
\end{theorem}

  Our rounding algorithm for the total throughput objective is very different from the current best approximation in \cite{GandhiKPS06} and other known approximations \cite{KimC04, ChekuriGIKLMMR10}. Let us call the interval $[r_\rho+1, d_\rho]$ request $\rho$'s  window.
We classify requests into two groups, depending on their window size. Our algorithm has two main  components. If small-window requests give a relatively large profit, we use a configuration LP to collect most profits from small window requests. Here by a configuration, we mean all possible transmissions that can be made during a small interval. The entire time horizon is partitioned into short disjoint intervals and configurations are defined for each of such intervals. The rounding is simply picking a configuration for each disjoint interval. Since we use configurations that capture enough details for small-window requests, we will be able to collect most profits from those, while we achieve an $(1 - 1/e)$-approximation for large-window requests.

The other component of our algorithm is used when large-window requests give a large profit. We modify the $\alpha$-point rounding in an interesting way. After the $\alpha$-point rounding, if multiple transmissions are made at a time, we keep only one transmission at random using the fair contention resolution scheme in \cite{FeigeV10} (interestingly, this already achieves an $(1 - 1/e)$ approximation), and let other transmissions walk either to the right or to the left at random for a certain constant number of time steps to find  an available empty time slot. More precisely, consider a transmission of page $p$ at time $t$ that is about to move to the left or to the right. Consider a large-window request $\rho$ such that $t \in [r_\rho +1, d_\rho]$ and $p_\rho = p$. Suppose that $t$ is fairly far from $d_\rho$. Then if the transmission moves to the right and can find an empty time slot soon, specifically by $d_\rho$, then the transmission will still satisfy the request $\rho$, and this is how we get more profits from large-window requests. It now remains to guarantee that there are enough empty time slots available so that the random walk of transmissions can find new places with a probability arbitrarily close to 1. To this end, after the $\alpha$-point rounding, we initially free away  an $\eps$-fraction of time slots, whose effect will be negligible to the approximation factor.

Although we use configurations that encode all possible transmission for an interval of length depending on $\eps$, note that the running time is in polynomial in 
$(1 / \eps)^{O(1 / \eps)}$ and the number of requests. This is achieved by solving the dual of the LP using an efficient separation oracle. 

\smallskip
Perhaps flow time and throughput  objectives are the most popular ones in the scheduling literature. To optimize flow times, it is crucial to minimize congestion during any interval since we are required to satisfy all requests, and congestion can accumulate over time, thereby increasing overall flow time considerably. On the other hand in the throughput objective the time constraints are hard but we are allowed to discard some requests. We believe that our variants of $\alpha$-point rounding that reduce congestion and resolve conflicts could be of potential use for other scheduling problems as well.

\subsection{Related Work}

A submodular generalization of the maximum throughput objective was studied in \cite{ChekuriGIKLMMR10}, motivated by some applications where each request gives a different profit depending on the time it is satisfied. In this extension, each request $\rho$ is associated with a submodular profit function that is defined over the times when page $p_\rho$ is transmitted. For this problem, \cite{ChekuriGIKLMMR10} gave a 0.5-competitive algorithm and a $(1-1/e)$-approximation algorithm. Other variants of this problem were considered and constant competitive algorithms were given \cite{ZhengFCCPW06, ChanLTW04, ChrobakDJKK06}. \cite{GandhiKPS06} gives a $0.75$-approximation for the maximum throughput objective in a slightly more general setting.

As mentioned before, the maximum flow time objective was first considered in \cite{BartalM00}. Chang \etal gave the first proof for the claim that $\fifo$ is 2-competitive for the objective \cite{ChangEGK11}. Later, Chekuri \etal extended the result to varying-sized pages using a different proof \cite{ChekuriIM12}. The performance guarantee of $\fifo$ is  tight \cite{ChangEGK11}, and it remains the case even if randomization is used to break ties between pages \cite{ChekuriGIKLMMR10}.  In the 2010 Dagstuhl seminar on scheduling, a special case for the throughput objective was introduced under the name of the IRS Tax Scheduling problem \cite{Dagstuhlscheduling10}. The problem, explained in broadcast scheduling terminology, asks if there is a feasible schedule that satisfies all requests $\rho$ during their window $[r_\rho +1,d_\rho]$ when there are at most two requests for each page. It was shown that an exact polynomial time algorithm for the Tax Scheduling problem yields a 1.75-approximation for the maximum flow time objective in broadcast scheduling (for all instances). It remains open if the Tax Scheduling problem is NP-hard or not.

Other interesting objectives were studied in broadcast scheduling. For the objective of minimizing the average flow time, the best approximation factor is $O(\log^2 n / \log \log n)$ \cite{BansalCS08}. We note that the algorithm and analysis in \cite{BansalCS08} can be used to give the same approximation guarantee for minimizing the maximum weighted flow time. In the online setting, there exists a $\Omega(n)$ lower bound on the competitive ratio for the average flow time objective, hence to overcome this lower bound, a relaxation called speed augmentation was used \cite{KalyanasundaramP95}. In this relaxation, the online algorithm runs on a machine that is $(1+\eps)$ times faster than the machine the optimal offline scheduler runs on. An online algorithm is said to be scalable if it is $f(\eps)$-competitive for any $\eps >0$. For average flow time and its variants, scalable algorithms were given \cite{ImM10, BansalKN10, EdmondsIM11}. There is a scalable algorithm known for the maximum weighted flow time objective \cite{ChekuriIM12}.

\subsection{Formal Problem Definition and Notation}

There are $n$ pages of information, and all pages are unit-sized. We will denote page mostly as $p$, $q$, and denote the set of pages as $\cP$.
Times are slotted. Requests can arrive only at a non-negative integer time. The server can transmit at most one page at each positive integer time. A request $\rho$ is released at time $r_\rho$ asking for a specific page $p_\rho$. 
We let $\cR$ denote the entire set of requests. We may add some condition as a subscript to denote a subset of requests that satisfy the condition. For example, $\cR_{r \geq t}$ will refer to the requests that arrive no earlier than time $t$. Let $m := |\cR|$ denote the number of requests. Without loss of generality, we assume that $m \geq n$. 
Also the set of times when  page $p$ is requested is denoted as $\cT_p$. 
We say that a schedule $\sigma$ is feasible if at most one page is transmitted at each time. In a schedule $\sigma$, request $\rho$ is completed at the first time $t' > r_\rho$ when the page $p_\rho$ is transmitted.
The flow time $F^\sigma_\rho := C^\sigma_\rho - r_\rho$ of  a request is the length of time that $\rho$ waits since its arrival until its completion. In the problem of  minimizing the maximum flow time, the goal is to find a feasible schedule that minimizes $\max_{\rho \in \cR} F^\sigma_\rho$. If the schedule $\sigma$ is clear from the context, it may be omitted. In the problem of maximizing the throughput (profit), each request $\rho$ has weight $w_\rho$ and deadline $d_\rho$. If page $p_\rho$ is transmitted during its window $W_\rho := [r_\rho +1, d_\rho]$, then we get a profit of $w_\rho$. The objective is to obtain a feasible schedule that gives the maximum total profit.

We let $T$ denote the maximum time we need to consider, which is at most $\max_{\rho \in \cR} r_\rho + n$. This is because transmitting all $n$ pages after the release of the last request completes all requests, and there is no incentive of transmitting the same page more than once when there are no more requests to arrive. We assume that all requests are given explicitly, and hence it follows that the input size is as large as $m$, the number of requests.
We suggest the reader to read Section~\ref{sec:mflow} and \ref{sec:throughput} assuming that $T$ is $O(m)$. Intuitively, if requests rarely arrive, at a rate of less than one request per time on average, the problem becomes easier. Later we will show how to remove this assumption.

\section{A PTAS for the Maximum Flow Time}
\label{sec:mflow}

In this section, we give a PTAS for the problem of minimizing the maximum flow time in broadcast scheduling. Consider any instance $J$ of requests. Let $L^*$ denote the maximum flow time of an optimal schedule on the instance. Note that $1 \leq L^*  \leq n$, since we can satisfy all requests within $n$ time steps by repeatedly transmitting all $n$ pages. We can without loss of generality assume that the quantity $L^*$ is known to the algorithm; this can be easily done by performing a binary search on the value of $L^*$ in the range of $[1, n]$. To make our algorithm and analysis more transparent, we will throughout assume that $T = O(m)$. In section~\ref{sec:fmax-removeT}, this simplifying assumption will be removed.

We observe that if $L^*$ is a constant, the problem can be solved in polynomial time using dynamic programming. We first give the key idea of our dynamic programming and then describe it in detail. Suppose we need to decide which page to transmit at time $t$ knowing the last $L^*$ transmissions made during $[t - L^*,  t-1]$ in the optimal schedule. Then in the remaining scheduling decision, we only need to care about $\cR_{r \geq t- L^*}$, the requests that arrive no earlier than $t  - L^*$. This is because all the other requests, $\cR_{r < t- L^*}$ must be satisfied by time $t -1$ in the optimal schedule. Since we do not know the ``correct" last $L^*$ transmissions, we simply enumerate all possible cases, whose number is at most $n^{L^*}$. Further, for each case, we keep track of if there exists a feasible schedule where all requests in $\cR_{r < t - L^*}$ are finished within $L^*$ time steps. Based on this observation, we propose the following dynamic programming.

Let $L$ denote our guess of $L^*$; we will simply consider $L$ in increasing order starting from 1. For $t \geq L-1$, let $\cQ(t) := \{ \{ (t - L +1, p_{t - L+1}),  (t - L +2 , p_{t - L+2}),  ...  ,(t, p_t)\} \; | \; p_{t - L+1}, p_{t - L+2}, ... , p_t \in \cP \}$ denote the collection of all possible transmissions that can be made during $[t - L+1, t]$; here a pair $(t' ,p_t')$ implies that page $p_{t'}$ is transmitted at time $t'$. Let us call an element $Q \in \cQ(t)$ as a configuration with respect to time $t$. Note that $|\cQ(t)| = n^{L}$. We say that a configuration $Q \in \cQ(t)$ is \emph{feasible} if there exists a schedule that is compatible with $Q$ where all requests in $\cR_{r \leq t - L}$ are completely satisfied by $t$ and have flow time at most $L$.

We are now ready to describe our algorithm. Let $\cQ^f(t)$ denote all feasible configurations in $\cQ(t)$. We will compute $\cQ^f(t)$ in increasing order of $t$ starting from $t = L-1$ to $t =T$. Note that $\cQ^f(L-1) = \cQ(L-1)$. To compute $\cQ^f(t)$ from $\cQ^f(t-1)$, we consider each pair of a configuration $Q \in \cQ^f(t-1)$ and a page $p_t \in \cP$. Let $Q  = \{ (t - L, p_{t - L}), (t - L+1 , p_{t - L+1}), ..., (t - 1, p_{t-1})\}$. We add $Q' = \{(t - L+1, p_{t - L+1}), (t- L+2, p_{t - L+2}), ..., (t, p_t) \}$ as a feasible configuration to $\cQ^f(t)$ if all requests with release time $t - L$ are satisfied by
$p_{t - L+1}, p_{t - L+2}, ..., p_{t - 1}$ or $p_t$. Finally, if $\cQ^f(T) = \emptyset$, then declare that there is no feasible schedule with the maximum flow time of at most $L$. Otherwise, it is straightforward to find a desired feasible schedule using the standard backtracking method, i.e. $\cQ^f(T), \cQ^f(T-1), \cQ^f(T-2), ..., \cQ^f(L-1)$. For completeness, we present the pseudocode of this algorithm. See Algorithm~\ref{algo-constant} in Section~\ref{sec:pseudo}.

The correctness of the algorithm hinges on the correctness of $\cQ^f(t)$, and it can be shown by a simple induction on $t$. Since the algorithm considers at most $T$ time steps, and at each time $t$ it considers all pairs of a configuration in $\cQ^f(t -1)$ and a page $p_t \in \cP$, whose total number is at most $n^{L} \cdot n$, we derive the following lemma (recall that $m \geq n$).

\begin{lemma}
    \label{lem:fmax-constant}
    There exists an optimal algorithm with run time $O(T) \cdot m^{O(L^*)}$ for the problem of minimizing the maximum flow time in broadcast scheduling.
\end{lemma}


Due to Lemma~\ref{lem:fmax-constant}, we can assume that $L^*$ is a sufficiently large constant.  We continue our analysis by considering two cases depending on the value  of $L^*$.  We begin with the case when $L^*$ is small.

\subsection{Case: $L^* \leq (1 / \eps^3) \log T$.}

Let $J$ denote the given instance of requests. We simplify $J$ to make our analysis easier. After the modification, we will distinguish time steps only to the extent that is sufficient to allow a PTAS. The simplification is described as follows: Shift each request $\rho$'s arrival time to the right to the closest integral multiple of $\eps L^*$, i.e. $\eps L^* \lceil \frac{r_\rho }{\eps L^*} \rceil$. Let $J'$ denote this modified instance. We will assume that $\eps L^*$ is an integer. This is justified by Lemma~\ref{lem:fmax-constant} which shows a polynomial time algorithm for any constant $L^*$. Assuming that $L^*$ is a sufficiently large constant, we can find an approximate value $L$ that is a multiple of $1/ \eps$ and multiplicatively close to $L^*$ in any arbitrary precision. The effect of this assumption on approximation factor and run time will be factored in later. Also we assume that $1/\eps$ is an integer.  For the modified instance $J'$, the scheduler is allowed to transmit only at times that are integral multiples of $\eps L^*$, and at each of such times, at most $\eps L^*$ pages. The following lemma shows that this modification almost preserves the maximum flow time and that a good schedule for $J'$ can be transformed into a good schedule for $J$. Recall that $F^\sigma_\rho$ denotes $\rho$'s flow time in schedule $\sigma$.

\begin{lemma}
    \label{lem:simplified-to-original}
    There exists a schedule $\sigma'$ for $J'$ with the maximum flow time of at most $(1 + 2 \eps) L^*$.  Further, any schedule $\sigma'$ for $J'$ can be converted into a feasible schedule $\sigma$ for the original instance $J$ such that the flow time of each request increases by at most $2\eps L^*$, i.e. $F^\sigma_\rho  - F^{\sigma'}_\rho \leq  2\eps L^*$ for all requests $\rho$.
\end{lemma}
\begin{proof}
    We first show the first claim.  Observe that shifting all transmissions in the optimal schedule to the right, to the closest multiple of $\eps L^*$ with additional $\eps L^*$ time steps, is a feasible schedule for $J'$.  Further this increases any request's flow time by at most $2 \eps L^*$. Thus the first claim follows. We now turn our attention to the second claim.  Recall that in $\sigma'$, at times $\eps L^* k $ for some integer $k$, at most $\eps L^*$ transmissions are made. Now move those transmissions to the times $\eps L^* k, \eps L^* k +1, \eps L^* k +2, ..., \eps L^* k + \eps L^*  - 1$. This may increase the flow time of each request by at most  $\eps L^*$. Since the arrival time of any request differs by at most $\eps L^*$ between two instances, $\sigma$ and $\sigma'$, the second claim follows.
\end{proof}

For notational simplicity, we will shrink the time horizon of the modified instance $J'$ by a factor of $\eps L^*$. Note that now in $J'$ requests can arrive at any positive integer time, and at most $\eps L^*$ transmissions can be made at each integer time. Hence the maximum flow time of the optimal schedule for $J'$ is at most $(2 + 1/\eps)$ by Lemma~\ref{lem:simplified-to-original}. We slightly modify Algorithm~\ref{algo-constant} (for constant $L^*$) to give a polynomial time algorithm for the modified instance $J'$ (for any fixed $\eps$). Throughout, let $\ell:=  2+ 1/\eps$. We first make a couple of useful observations. Let $\cP_t$ denote the pages requested at time $t$, and $\cA_t$ the pages transmitted by our algorithm at time $t$.

\begin{proposition}
    \label{prop:smallL-1}
    There exists an optimal schedule $\cA_t$ for $J'$ with the maximum flow time $\ell$ such that
    \begin{itemize}
        \item for any page $p$ and time $t$, at most one transmission is made for page $p$ at time $t$.
        \item for any interval $I$ of length $\ell$ and  for any page $p$, there are at most two transmissions made for  page $p$ during $I$.
        \item for any time $t$, $\cA_t \subseteq \bigcup_{i = 1}^{\ell} \cP_{t - i}$.
    \end{itemize}
\end{proposition}
\begin{proof}
    Obviously there is no incentive of broadcasting a page more than once at the same  time. Suppose there is an interval of length $\ell$ during which page $p$ is transmitted more than twice. Then we can remove a middle transmission while keeping the maximum flow time no greater than $\ell$. The final claim follows since if it were not true, any transmission of page $p$ that was not requested for the last $\ell$ time steps would be wasted.
\end{proof}

\begin{proposition}
    \label{prop:smallL-2}
    For any $t$, the number of pages that could  be potentially transmitted   at time $t$ for any feasible schedule of the modified instance with maximum flow time $\ell$    
    is at most $|  \bigcup_{i = 1}^{\ell} \cP_{t - i}| \leq 6 L^*$. 
    Further, the number of pages that could be potentially transmitted during $[t - \ell, t-1]$ for any feasible schedule of the modified instance with maximum flow time $\ell$ is at most $|\bigcup_{i = 2}^{2 \ell} \cP_{t - i}| \leq 9 L^*$.    
\end{proposition}
\begin{proof}
    All pages in $\bigcup_{i = 1}^{\ell} \cP_{t - i}$ must be transmitted during $[t - \ell +1, t + \ell -1]$. Hence the total number of such pages is at most $2 \ell * \eps L^* = 2 \eps L^*( 2 + 1/ \eps) \leq 6 L^* $.  The last property in Proposition~\ref{prop:smallL-1} completes the proof of the first claim. 
        The second claim can be proven similarly. By the last property in Proposition~\ref{prop:smallL-1}, we have that $\bigcup_{i = 1}^{\ell} \cA_{t-i} \subseteq  \bigcup_{i = 2}^{2 \ell} \cP_{t - i}$. Since all pages in $\bigcup_{i = 2}^{2 \ell} \cP_{t - i}$ must be transmitted during $[t -2 \ell +1, t + \ell -1]$, we conclude that $|\bigcup_{i = 2}^{2 \ell} \cP_{t - i}| \leq 9 L^*$.
\end{proof}

For any $t \geq \ell - 1$, we say that  $\{(t - \ell +1, \cA_{t - \ell +1} ), (t - \ell +2, \cA_{t - \ell +2} ), ..., (t , \cA_{t} ) \} $ is a configuration with respect to time $t$ if for all $ 0 \leq i <  \ell,  \cA_{t - i} \subseteq \cP  \mbox{ and }  | \cA_{t-i}|  \leq \eps L^*$. In other words, a configuration represents all transmissions made during $[ t - \ell +1, t ]$.  Here $\cA_t$ can be any subset of at most $\eps L^*$ pages requested in the last $\ell$ time steps. Let $\cQ(t)$ denote the collection of all configurations with respect to time $t$. Further we say that $Q \in \cQ(t)$ is feasible if there exists a schedule compatible with $Q$ where $\cR_{r \leq t - \ell}$, all requests arriving no later than $t - \ell$ are completed by time $t$ and have flow time at most $\ell$. Let $\cQ^f(t)$ denote the collection of all feasible configurations in $\cQ(t)$.

We are now ready to describe our algorithm. This algorithm is very similar to Algorithm~\ref{algo-constant} for constant $L^*$. We compute $\cQ^f(t)$ in increasing order of time $t$ starting from $\ell$ to $T' := \lceil \frac{\max_\rho r_{\rho}}{ \eps L^*} \rceil + \ell$. Note that $\cQ^f(\ell- 1) = \cQ(\ell- 1)$. To obtain $\cQ^f(t)$ from $\cQ^f(t-1)$, we consider each pair of $Q \in \cQ^f(t - 1)$ and a potential set of pages $\cA_t$. Let $Q = \{ (t - \ell, \cA_{t - \ell}), (t - \ell+1 , \cA_{t - \ell+1}), ..., (t - 1, \cA_{t-1})\} \in \cQ^f(t - 1)$. Here $\cA_t$ satisfy (1) all requests with release time $t - \ell$ are satisfied by pages in $\bigcup_{i=0}^{\ell-1} \cA_{t - i}$, and (2) no page appears more than twice in $\cA_{t - \ell+1}, \cA_{t - \ell +2}, ..., \cA_t$ (The second condition is justified by the second property in Proposition~\ref{prop:smallL-1}). If it is the case, then we add $Q' = \{(t - \ell+1, \cA_{t - \ell+1}), (t- \ell+2, \cA_{t - \ell+2}), ..., (t, \cA_t) \}$  to $\cQ^f(t)$. At the end of the algorithm, if $\cQ^f(T') = \emptyset$, we declare that any feasible schedule has the maximum flow time larger than $\ell$. Otherwise, using $\cQ^f(\cdot)$, we construct a feasible schedule with the maximum flow time of at most $\ell$. The pseudocode of this algorithm can be found in Algorithm~\ref{algo-small-2} in Section~\ref{sec:pseudo}. The correctness of this algorithm can be easily shown by a simple induction on $t$.


\begin{lemma}
    If there exists a schedule for $J'$ with the maximum flow time $\ell$, Algorithm~\ref{algo-small-2} yields such a schedule. Otherwise, it declares that the maximum flow time is greater than $\ell$.
\end{lemma}

We now upper-bound the run time of this algorithm.

\begin{lemma}
    \label{lem:fmax-small-runtime}
    Algorithm~\ref{algo-small-2} has a run time of at most $O\left(\frac{T}{ \eps L^*} \cdot \left(L^*\right)^2 \right) \cdot \left(\frac{3}{\eps}\right)^{18L^*} \cdot \left(\frac{6e}{\eps}\right)^{\eps L^*}$.
\end{lemma}
\begin{proof}
    We start by upper-bounding several quantities. We first bound the size of $\cQ^f(t-1)$.
    $$|\cQ^f(t-1)| \leq |\cQ(t-1)| \leq (\ell^2)^{9L^*} \leq \left(\frac{3}{\eps}\right)^{18L^*}$$
    This follows since by Proposition~\ref{prop:smallL-2}, there are at most $9L^*$ pages that can appear in a configuration in $\cQ^f(t-1)$ and by Proposition~\ref{prop:smallL-1}, each of those pages can be transmitted at most twice during $[ t- \ell, t-1]$.

    By Proposition~\ref{prop:smallL-2} , the total number of different sets $\cA_t$  is upper-bounded by
    $${6 L^* \choose \eps L^*} \leq \left(\frac{6e L^*}{\eps L^*}\right)^{\eps L^*} = \left(\frac{6e}{\eps}\right)^{\eps L^*} $$
    It is easy to see that the extra overhead for each pair of $Q$ and $\cA_t$ is at most $O((L^*)^2)$ (we do not optimize this, since this is negligible compared to the above two quantities),  hence the run time follows.
\end{proof}

Now recall that we have assumed that $\eps L^*$ is an integer. Suppose $\eps L^*$ is not an integer. If $L^* \leq 1 / \eps^2$, by Lemma~\ref{lem:fmax-constant}, we can find an optimal schedule in time $O(m) \cdot n^{O( 1/\eps^2)}$. If $L^* \geq 1/ \eps^2$, then we can find $L$ such that $L^* < L < L^* + 1/\eps$ and $L$ is an integral multiple of $1/\eps$. Since $L$ approximates $L^*$ within a factor of $1 + \eps$, by considering $L$ rather than $L^*$, we will lose only a multiplicative factor $1+ \eps$ in the approximation ratio.

By Lemma~\ref{lem:simplified-to-original}, we have the following theorem.

\begin{theorem}
    \label{thm:smallL}
     For any $0 < \eps \leq 1$, there exists a $(1 + 6\eps)$-approximation for the maximum flow time in broadcast scheduling with run time $O\left(\frac{T}{ \eps L^*} \cdot (L^*)^2 \right)\cdot \left(\frac{3}{\eps}\right)^{18L^*} \cdot \left(\frac{6e}{\eps}\right)^{\eps L^*}$.
\end{theorem}

\subsection{Case: $L^* \geq (1 / \eps^3) \log T$.}
    \label{sec:largeL}

This section is devoted to proving the following lemma.

\begin{lemma}
    \label{cor:fmax-random}
    Suppose that $L^* \geq \frac{1}{\eps^3} \log T$. Then there exists a randomized algorithm that yields a feasible schedule with the maximum flow time of at most $(1 +6\eps)L^*$ with a probability of at least $1 - 1/ T$.
\end{lemma}

This lemma and Theorem~\ref{thm:smallL} would yield a randomized version of one of our main results.

\begin{theorem}
    \label{thm:fmax-main-random}
    There exists a randomized PTAS for minimizing the maximum flow time in broadcast scheduling. More precisely, there exists a randomized $(1 + \eps)$-approximation with run time $m^2 \cdot T^{O(1/ \eps^4)}$ that succeeds with a high probability.
\end{theorem}
\begin{proof}
    The run time is upper-bounded  by the quantity  $O\left(T \cdot (L^*)^2\right) \cdot \left(\frac{3}{\eps}\right)^{18L^*} \cdot \left(\frac{6e}{\eps}\right)^{\eps L^*}$ with $L^* = (1 / \eps^3) \log T$, which simplifies to
    \begin{eqnarray*}
        && O(T \cdot m^2) \cdot \left(\frac{3}{\eps}\right)^{(18/ \eps^3) \log T} \cdot \left(\frac{6 e}{\eps}\right)^{(1/ \eps^2) \log T} \leq O(T \cdot m^2)  \cdot \left(\frac{3}{\eps}\right)^{(18/ \eps^3) \log T} \cdot \left(\frac{3}{\eps^3}\right)^{(3/ \eps^2) \log T} \\
        &\leq& O(T \cdot m^2)\cdot \left(\frac{3}{\eps}\right)^{(21/ \eps^3) \log T} = O(T \cdot m^2) \cdot T^{\frac{21}{\eps^3} \log \frac{3}{\eps}}   = m^2 \cdot T^{O( \frac{1}{\eps^4}) }
    \end{eqnarray*}
\end{proof}

As mentioned before, we will remove the dependency of run time on $T$ in Section~\ref{sec:fmax-removeT}. Also in Section~\ref{sec:derandomization}, using the method of pessimistic estimators, we will derandomize the rounding scheme described in this section, and complete the proof of Theorem~\ref{thm:fmax-main-det}. We note that one can improve the run time of our randomized algorithm to $m^2 \cdot T^{O( 1/ \eps^3)}$ by considering two cases that $L^* \geq ( 1/ \eps^2) \log T$ or not. However, the corresponding derandomization process seems more tricky and involved, hence we present the analysis of  a slightly worse run time.

\medskip

We will now focus on proving Lemma~\ref{cor:fmax-random}. For notational convenience, we begin with simplifying the given instance of requests. We say that two requests $\rho$ and $\rho'$ for the same page $p$ are adjacent if there are no requests for page $p$ released between time $r_{\rho}$ and $r_{\rho'}$. We claim that we can without loss of generality assume that any two adjacent requests for the same page arrive within less than $L^*$ time steps. To see this, suppose that there are two adjacent requests $\rho$ and $\rho'$ for page $p$ that arrive apart by at least $L^*$ time steps. Then in the optimal schedule, no transmission of page $p$ can be used to satisfy $\rho$ and $\rho'$ simultaneously. Hence we can assume that the requests in $\cT_{p, r \leq r_\rho}$ and those in $\cT_{p, r \geq r_{\rho'}}$ are for different pages, where $\cT_{p, r \leq r_\rho}$ and $\cT_{p, r \geq r_{\rho'}}$ denote the set of requests for page $p$ that arrive no later than $\rho$ and no earlier than $\rho'$, respectively.

Henceforth, we assume that any two adjacent requests $\rho$, $\rho'$ (for the same page) arrive within less than $L^*$ time steps, i.e. $|r_\rho - r_{\rho'}| \leq L^* -1$.
We consider the following integer programming that determines if there is a feasible schedule with the maximum flow time of at most $L^*$.
\begin{alignat}{10} \tag{$\mathsf{LP}_\mathsf{MaxFlow}$} \label{lp:mflow}\\
\sum_{t' = t+1}^{t+L^*} x_{p,t'}  &\geq 1  \qquad \qquad  \qquad  & \forall p \in \cP, t \in \cT_p \label{LP-mflow-1} \\
\sum_{p \in \cP} x_{p,t} &\leq 1 &\forall t \in [T] \label{LP-mflow-2}\\
             x_{p,t}  &= 0 &\forall p \in \cP, t \notin [\min \cT_p +1 , \max \cT_p + L^*] \label{LP-mflow-3}\\
             x_{p,t}  &\in \{0, 1\} &\forall p \in \cP, t \in [T] \nonumber
\end{alignat}

The first constraints say that all requests must be completed within $L^*$ time steps. The second constraints state that at most one page can be transmitted at each time. The third constraints ensure that no transmission is made if it cannot be used to satisfy a request within $L^*$ time steps. We relax the integer programing by replacing the last constraints with $x_{p, t} \geq 0$.
We solve $\lpmflow$ and let $x^*_{p,t}$ denote the optimal solution. Our rounding scheme will be based on the groups we define as follows. For each page $p$, define $W_p:= [\min \cT_p +1, \max \cT_p + L^*]$, which we call page $p$'s window. Note that there exists an optimal schedule where page $p$ is broadcast only during $W_p$. We now partition pages into groups $\cG$ such that all pages in the same group $g \in \cG$ have disjoint windows. We show that $2L^*$ groups suffice for this partition.

\begin{lemma}
    \label{lem:fmax-num-groups}
    We can in polynomial time partition pages $\cP$ into at most $2L^*$ groups in $\cG$ such that for any $p \neq q$ in the same group $g$, $W_p$ and $W_q$ are disjoint.
\end{lemma}
\begin{proof}
We create an interval graph graph $G$ with one interval corresponding to each window $W_p$. The partition of pages into groups corresponds to a feasible vertex coloring of $G$. It is well known that the chromatic number of an interval graph is equal to its clique number. Moreover, such a coloring can be easily found by a greedy algorithm.

We claim that a clique number of $G$ is upper bounded by $2L^*$. To show this we prove that at any time $t$,  there are at most $2L^*$ windows $W_p$ that contain $t$. To this end, it suffices to show that if $t \in W_p$, then $p$ must be transmitted at least once during $[t - L^* +1 , t + L^*]$. We consider two cases. If $t \geq \max \cT_p$, then this claim holds since to satisfy the last request, page $p$ must be transmitted during $[ \max \cT_p +1,  \max \cT_p + L^*]$, which is contained in $[t - L^* + 1 , t + L^*]$. Otherwise, there must exist a request for $p$ that arrives at some time $t' \in  [t - L^*+1, t]$, and to satisfy the request, page $p$ must be transmitted at least once during $[t - L^* +2, t + L^*]$.
\end{proof}

We now describe our rounding scheme. For each group $g \in \cG$, we define a cumulative amount of transmission made by  the fractional solution. Formally, define $y^*_{g, t} := \sum_{p \in g} \sum_{t' \leq t}  x^*_{p,t'}$. Now for each group $g$, pick $\alpha_g$ from $[0,1]$ uniformly at random, and this is the only random value for group $g$. We first obtain a tentative schedule $\stemp$ and then the final schedule $\sfinal$. In the tentative schedule $\stemp$, all requests are satisfied within $L^*$ time steps, but it is allowed to transmit more than one page at a time. In $\stemp$, for each group $g$, we transmit a page $p$ in $g$ at times $t$ such that $y^*_{g,t-1} <  k + \alpha_g \leq y^*_{g,t}$ for some integer $k$. Here the page $p \in g$ transmitted at time $t$ for group $g$ is such that $x^*_{p,t} >0$. Note that for any time $t$, there is at most one page $p \in g$ with $x^*_{p,t} >0$. This is due to the definition of groups and the constraints (\ref{LP-mflow-3}).  We now transform $\stemp$ into the final feasible schedule $\sfinal$ in the First-In-First-Out fashion: Think of a transmission of page $p$ at time $\tau$ as a job $j(p,\tau)$ that arrives at time $\tau$, and add the job to the queue. At each time $t$, we dequeue the job $j(p,\tau)$ with the earliest arrival time $\tau$ and transmit page $p$. Clearly, in the final schedule $\sfinal$, at most one transmission is made at each time.  We will show that in $\sfinal$, all requests have flow time at most $(1 + 6\eps) L^*$ with high probability.

We begin our analysis with the following easy proposition concerning the tentative schedule $\stemp$.

\begin{proposition}
    In $\stemp$, all requests have flow time at most $L^*$.
\end{proposition}
\begin{proof}
    Consider any request $\rho$ for page $p$, and let $g \in \cG$ be the group that contains page $p$. Then due to constraints (\ref{LP-mflow-1}), we have $y^*_{g, r_\rho + L^*}  - y^*_{g, r_\rho} = \sum_{\tau =r_\rho+1}^{r_\rho+L^*} x^*_{p,t} \geq 1$. Hence for any random value $\alpha_g$, a page is transmitted from $g$ during $[r_\rho+1, r_\rho + L^*]$. Further, since $[r_\rho +1, r_\rho + L^*] \in W_p$ and all pages in $g$ have disjoint windows, page $p$ is transmitted during $[r_\rho +1, r_\rho + L^*]$.
\end{proof}

We complete the analysis by showing that w.h.p. each request's flow time does not increase too much in the transformation from $\stemp$ into $\sfinal$. To this end, we need to measure the number of transmissions made in $\stemp$ during an interval $I$. Let $\qtemp(I)$ denote this quantity. We need the following simple lemma, which was shown in \cite{BansalCS08}.
The lemma shows how the quantity  $\qtemp(I)$ is related to the increase of each request's flow time from $\stemp$ to $\sfinal$.

\begin{lemma}
    \label{lem:fmax-delay-2}
    For all  requests $\rho$, $\Ffinal_\rho \leq \Ftemp_\rho + \max_{1 \leq t_1 \leq t_2 \leq T} \max \{ \qtemp( [t_1, t_2] )  - (t_2 - t_1 +1), 0 \}$.
\end{lemma}
\begin{proof}
    Consider any request $\rho$. Let $t_0$ denote the latest time $t < r_\rho$ such that there is only one transmission made at time $t$ in $\stemp$ and the transmission is scheduled exactly at time $t$ also in $\sfinal$ without any delay, or no transmission is made at time $t$ in $\sfinal$. If no such time $t$ exists, then $t_0= 0$. Note that $\Ffinal_\rho - \Ftemp_\rho \leq (\qtemp( [t_0 +1 , r_\rho] )  - (r_\rho - ( t_0 +1 ) + 1)$. Hence the lemma follows.
\end{proof}

Motivated by the above lemma, we define the overflow during $[t_1, t_2]$,
$$ \overflow([t_1, t_2]) := \max \{ \qtemp( [t_1, t_2] )  - (t_2 - t_1 +1), 0 \}.$$

Formally,  we will upper-bound the overflow during an interval as follows. By a simple union bound over all possible intervals (at most $T^2$), it will imply Lemma~\ref{cor:fmax-random}.

\begin{lemma}
    \label{lem:fmax-delay}
    Suppose that $L^* \geq (1 / \eps^3) \log T$. Then for any interval $I$, $\Pr[\overflow(I) \geq 6\eps L^*] \leq 1/ T^3$.
\end{lemma}
\begin{proof}
 Consider any fixed interval $I = [t_1, t_2]$. For each group $g \in \cG$, let $N_g$ denote the number of transmissions that are made for pages in $g$ during $I$ in $\stemp$, which is of course a random variable. Let $v_g: = \sum_{p \in g} \sum_{t \in I} x^*_{p,t}$ denote the volume that $\lpmflow$ transmits for the pages in $g$ during $I$. Note that $\Ex [ N_g] = v_g$.
 Also observe that $X_g := N_g - \lfloor v_g \rfloor$ is a 0-1 random variable.  This is because a page in $g$ is transmitted in $\stemp$ at every time step $\lpmflow$ makes one additional volume of transmission for pages in $g$ (except the first one which is done when  $\alpha_g$ volume of transmission is done by $\lpmflow$). Note that $X_g$ are independent with  $\Pr[ X_g] =  v_g -  \lfloor v_g \rfloor$. By  Lemma~\ref{lem:fmax-num-groups}, we know that there are at most $2L^*$ random variables $X_g$. For notational convenience, let $\mu_g := \Ex[X_g]$ and $\mu := \sum_{g \in \cG} \mu_g$. 

We can relate the probability in the claim to the following probability of an event in terms of $X_g$. 
\begin{eqnarray*}
    \Pr \Big[\overflow(I)  \geq 6 \eps L^* \Big] &\leq& \Pr \Big[ \sum_{g \in \cG} N_g - \sum_{g \in \cG} v_g \geq 6 \eps L^* \Big] \\
    &=& \Pr \Big[ \sum_{g \in \cG} (N_g - \lfloor v_g \rfloor) \geq  \sum_{g \in \cG} (v_g - \lfloor v_g \rfloor) + 6 \eps L^* \Big] \\
    & =& \Pr \Big [ \sum_{g \in \cG} X_g \geq \mu + 6 \eps L^* \Big]
\end{eqnarray*}

The first inequality follows from the facts that $Q_{temp}(I) = \sum_{g \in \cG} N_g$ and $\sum_{g \in \cG} v_g = \sum_{g \in \cG}  \sum_{p \in g} \sum_{t \in I} x^*_{p,t} \leq t_2 - t_1 +1$. The last equality comes from the definition of $X_g$. Knowing that $\Var [ \sum_{g \in \cG} X_g ] \leq 2L^*$, by applying Theorem~\ref{thm:bernstein} (Bernstein Concentration Inequality) with $\Delta = 6\eps L^*$, $V \leq 2L^*$ and $b \leq 1$, we derive
$$\Pr \Big[ \sum_{g \in \cG} X_g \geq \mu + 6 \eps L^*  \Big] \leq \exp \Big( - \frac{\Delta^2}{2V + (2/3)b \Delta} \Big ) \leq \exp \Big( - \frac{( 6\eps L^*)^2}{4L^* + 4 \eps L^*}\Big) < \exp (- 3 \eps^2 L^*) \leq \frac{1}{T^3}$$
\end{proof}
%
%
%

\subsection{Removing the Dependency of Run Time on $T$}
\label{sec:fmax-removeT}

In this section, we will make our algorithm run in polynomial time in $m$ for any fixed $\eps >0$. We consider three cases depending on the value of $L^*$, and show how the proposed algorithm for each case can be adapted so that the run time does not depend on $T$. Recall that $T$ is the length of the planning horizon, i.e. an upper-bound on the last time when the optimal schedule transmits. Let $r_{\max} := \max_\rho r_\rho$ denote the latest arrival time of any request. One can without loss of generality set $T := r_{\max} + \min \{m,n\}$.
This is because one can satisfy all requests by $T$ by transmitting all outstanding pages after $r_{\max}$ and there is no incentive of transmitting the same page twice when there are no more requests to arrive.  Hence if $T = O(m)$, we have nothing to prove. So we will assume that $T > m$. We will say that an interval $I$ is silent if no request arrives during $I$.
 \\

\noindent
\emph{Case (a): $L^* = O(1)$.}  Recall that Algorithm~\ref{algo-constant} runs in $O( T) \cdot m^{O(L)}$. Suppose that there is a silent interval $I  = [t_1, t_2] \subseteq [1, T]$ of length at least $L^*$.  Suppose that $I$ is maximal. Then since all requests in $\cR_{r < t_1}$ must be satisfied by time $t_1 + L^* -1 (\leq t_2)$, one can find an optimal schedule on the requests in $\cR_{r < t_1}$ and $\cR_{r > t_2}$ separately. Hence one can without loss of generality assume that there is no silent interval of length $L^*$. This implies that $T = O(L^* \cdot m) = O(m^2)$. Hence we obtain an exact algorithm with run time $m^{O(L^*)}$.

\smallskip
\noindent
\emph{Case (b): $L^* \leq \frac{1}{\eps^3} \log m$.} This case is similar to case (a). In Theorem~\ref{thm:smallL}, we have shown a $(1 + 6\eps)$-approximation with run time  $O(\frac{T}{ \eps L^*} (L^*)^2) \cdot (\frac{3}{\eps})^{18L^*} \cdot (\frac{6e}{\eps})^{\eps L^*}$; here $O(\frac{T}{ \eps L^*})$ is due to the number of time steps considered in the modified instance $J'$. Recall that in $J'$, the maximum flow time $\ell$ was at most $ 2 + 1/ \eps$. By a similar argument as for case (a), we can assume that there is no silent interval of length $\ell$. Hence it follows $\frac{T}{ \eps L^*}  = O(\ell m)$. Hence we obtain a $(1 + 6\eps)$-approximation with run time  $O( \frac{m}{\eps}) \cdot (\frac{3}{\eps})^{18L^*} \cdot (\frac{6e}{\eps})^{\eps L^*}$.

\smallskip
\noindent
\emph{Case (c): $L^* \geq \frac{1}{\eps^3} \log m$.} We will show that one can modify $\mathsf{LP}_\mathsf{MaxFlow}$ so that the number of time steps (not necessarily continuous) in consideration is at most $m$. This will allow us to find an optimal LP solution where $\sum_{p \in \cP} x^*_{p, t} \neq 0$ for at most $m$ time steps $t$.  We will show that one can decompose the requests instance into a sub-instance where the number of time steps is as large as the number of requests, and solve each sub-instance separately. We assume without loss of generality that $\cR_{r =0} \neq \emptyset$.
Suppose that $T > m$, since otherwise there is nothing to prove. Let $t_2 \geq 0$ be the earliest time $t$ such that $| \cR_{0 \leq r \leq t}| \leq t+1$ (or equivalently, $| \cR_{0 \leq r \leq t}| = t+1)$.

\begin{claim}
    The optimal schedule (in fact any reasonable schedule that tries to satisfy at least one outstanding request) satisfies all requests $\cR_{r \leq t_2}$ by time $t_2 +1$.
\end{claim}
\begin{proof}
Without loss of generality assume that the optimal solution transmits a page only if it satisfies at least one outstanding request. If the optimal schedule is never idle by time $t_2 +1$ (except at time 0), then the claim holds, since $| \cR_{0 \leq r \leq t_2}| = t_2+1$. Also if the optimal schedule is idle at time $t_2+1$, then the claim again follows, since there are no outstanding requests at time $t_2+1$. If not, we show a contradiction. Let $t_1$ be the latest time no greater than $t_2$ when the optimal schedule gets idle.  Note that  $|\cR_{t_1 \leq r \leq t_2}| < t_2 -(t_1-1)$. Otherwise, we will have $|\cR_{0 \leq r \leq t_1 - 1}| \leq  t_1 $, which is a contradiction to the definition of $t_2$. Hence all requests in $\cR_{t_1 \leq r \leq t_2}$ are satisfied during $[t_1+1, t_2 +1]$. Then it implies that the optimal schedule makes $(t_2 - t_1 +1)$ transmissions to satisfy $t_2 - t_1$ requests. A contradiction.
\end{proof}

We have shown that all requests in $\cR_{0 \leq r \leq t_2}$ ($| \cR_{0 \leq r \leq t_2}| = t_2 +1$) can be satisfied during $[1, t_2 +1]$. This will become one sub-instance. By repeating this process on the remaining requests $\cR_{r > t_2}$, we can identify at most $m$ time steps that we need to consider.  Let $\cT'$ denote the set of such times. Then we restrict the LP variables only to times in $\cT'$. Recall that in Lemma~\ref{lem:fmax-delay} we showed that $\Pr[\overflow(I) \geq 6\eps L^*] \leq \frac{1}{T^3}$. The proof can be easily modified to show that when $L^* \geq \frac{1}{\eps^3} \log m$, $\Pr[\overflow(I) \geq 6\eps L^*] \leq \frac{1}{m^3}$. We claim that we only need to focus on the intervals $I$ that start and end at times in $\cT'$. To see this, consider any interval $[\tau_1, \tau_2]$. Let $[\tau_1', \tau_2']$ be the maximal interval such that $\tau_1', \tau_2' \in \cT'$, and $[\tau_1', \tau_2'] \subseteq [\tau_1, \tau_2]$; if no such $\tau_1'$, $\tau_2'$ exist, then $\overflow( [ \tau_1, \tau_2]) = 0$, since in $\stemp$, transmissions can be made only at times in $\cT'$. Note that $\overflow([ \tau_1', \tau_2']) \geq \overflow([ \tau_1, \tau_2])$. Since there are at most $m^2$ intervals to consider, by a simple union bound, we obtain a randomized $(1 + 6\eps)$-approximation that succeeds with a probability of at least $1 - 1/m$.

\medskip
From the above three cases, we can show that the run time is $O( \frac{m}{\eps}) \cdot (\frac{3}{\eps})^{18L^*} \cdot (\frac{6e}{\eps})^{\eps L^*}$ with $L^* = (1 / \eps^3) \log m$. By  a similar algebra as in the proof of Theorem~\ref{thm:fmax-main-random}, the run time simplifies to $m^{O( 1/ \eps^4)}$, thereby removing the dependency on $T$.

\subsection{Derandomization}
    \label{sec:derandomization}

\newcommand{\cIB}{\cI^{big}}
\newcommand{\cIS}{\cI^{small}}

We derandomize our algorithm in Section~\ref{sec:largeL} using the method of pessimistic estimators \cite{Raghavan88}. This method is by now a standard tool for derandomizing an algorithm that relies on  concentration  inequalities. It is useful particularly when there is no concrete target value whose expectation is easy to measure, and hence the standard conditional expectation method can not be used.  Recall from Section~\ref{sec:largeL} that for each interval $I \in \cI$, we created a set of 0-1 variables $X_{g, I}$ with $\Pr [ X_{g, I} = 1] = v_{g, I}  - \lfloor v_{g, I} \rfloor$, and showed that $\Pr \Big[ \sum_{g \in \cG} X_{g, I} \geq \Ex [\sum_{g \in \cG} X_{g, I}] + \Delta \Big] \leq \frac{1}{m^3}$, where $\Delta = 6 \eps L^*$. Let $\mu_{g, I}:=  v_{g, I}  - \lfloor v_{g, I} \rfloor$ and $\mu_I = \sum_{g \in \cG} \mu_{g, I}$. (In Section~\ref{sec:largeL}, to simplify the notation, we considered any fixed interval $I$ and did not include $I$ in the subscript. Also, initially we showed the probability is bounded by $1 / T^3$, and in Section~\ref{sec:fmax-removeT} that we only need to consider at most $m^2$ intervals and the probability is bounded by $1 / m^3$). Recall that for any $g \in \cG$,  all random variables $\{ X_{g,I} \}_{I \in \cI}$ are determined by the same random value $\alpha_g \in [0,1]$. For notational convenience, we index groups by integers from 1 to $k := |\cG|$ in an arbitrary but fixed way and will refer to groups by their index. To apply the method of pessimistic estimators, we define our ``estimator" for each $I \in \cI$. We will distinguish intervals in $\cI$ into two groups: If $\mu_I \geq \eps L^*$ then $I \in \cIB$, otherwise $I \in \cIS$. Throughout this section, let $\lambda := 1+ 3\eps$.

For each $I \in \cIB$, define
$$f_I(  z_1, z_2, ..., z_k):= \frac{ \prod_{i \in [k]}  \exp  \Big ((\log \lambda) z_i \Big)  }{\exp \Big( \lambda (\log \lambda) \mu_I\Big)}$$
and for each $I \in \cIS$, define
$$f_I(  z_1, z_2, ..., z_k):= \frac{ \prod_{i \in [k]}  \exp  ( z_i )  }{\exp (6 \eps L^*)}$$
We will consider groups in increasing order of their index, and will find a ``right" $\alpha_i$ value conditioned on the $\alpha_{1}, \alpha_{2}, ..., \alpha_{i-1}$ values that we have found and fixed. To this end, we need to  carefully define several quantities. Let $x_{i, I}(\alpha'_{i})$ denote $X_{i,I}$'s value when $\alpha_i = \alpha'_i$. When $\alpha'_i$ is clear from the context, we may denote $x_{i, I}( \alpha'_{i})$ simply by $x'_{i, I}$. When $\alpha_1= \alpha'_1$, ..., $\alpha_{i} = \alpha'_{i}$ are fixed, we define the following quantity for each interval $I \in \cI$.
$$E_{i, I} := \Ex[ f_I( X_1, X_2, ..., X_k) \; | \; \alpha_1 = \alpha'_1, \alpha_2 = \alpha'_2, ..., \alpha_{i} = \alpha'_{i} ], $$
where the expectation is defined over $X_{i+1}, X_{i+2}, ..., X_{k}$ (which are determined by $\alpha_{i+1}, \alpha_{i+2}, ..., \alpha_k$). We note that $E_{i,I}$ can be easily computed in polynomial time since $X_{i, I}, i \in [k]$ are independent.

We will show the following two lemmas.
\begin{lemma}
    \label{lem:derand-1}
    Consider any integer $0 \leq h \leq k$, and any $x_{i, I} \in \{0, 1\}$, $i \in [h]$ (more precisely for any $\alpha_{i, I} \in [0,1]$, $i \in [h]$ that determines $x_{i,I}$, $i \in [h]$). Then for any $I \in \cIB$,
    $$\Pr \Big[ \sum_{i \in [k]} X_{i, I} \geq \lambda \mu_I \; \Big | \; X_{i,I}  = x_{i,I},  i \in [h]  \Big] \leq  E_{h, I}.$$
    Also for any $I \in \cIS$,
    $$\Pr \Big[ \sum_{i \in [k]} X_{i, I} \geq 6 \eps L^*\; \Big | \; X_{i,I}  = x_{i,I},  i \in [h]  \Big] \leq  E_{h, I}.$$
\end{lemma}

\begin{lemma}
    \label{lem:derand-2}
    $\sum_{I \in \cI} E_{0,I} \leq \frac{1}{m}$.
\end{lemma}

Lemma~\ref{lem:derand-1} shows how the probability of the bad event for $I$ (an overflow more than $6\eps L^*$) can be bounded by the expectation of our estimator function. Furthermore, this inequality holds for the bad event for $I$ conditioned on any fixed $X_{i,I}$, $i \in [h]$ (or more precisely on fixed $\alpha_{i}$, $i \in [h]$). This will allow us to use the conditional expectation method to find a sequence of $\alpha_1 = \alpha'_1$, $\alpha = \alpha'_2$, ..., $\alpha_k= \alpha'_k$ such that
$$\sum_{I \in \cI} E_{ k, I}  \leq \sum_{I \in \cI} E_{ k-1, I} \leq ... \leq \sum_{I \in \cI} E_{0, I}$$

By Lemma~\ref{lem:derand-2}, we know that $\sum_{I \in \cI} E_{ k, I} \leq 1/ m$. Knowing that $\mu_I \leq 2L^*$, it is easy to see that  $\alpha'_1, \alpha'_2, ..., \alpha'_k$ are the desired ``good" $\alpha$ values such that for all $I \in \cI$,
    $$\sum_{i \in [k]} x_{i, I}(\alpha'_i) \leq \mu_{I}+ 6 \eps L^* .$$

\medskip
It now remains to prove Lemma~\ref{lem:derand-1} and \ref{lem:derand-2}. The proofs are very similar to that of  Chernoff inequalities.

\begin{proofof}[Lemma~\ref{lem:derand-1}]
    We first consider any $I \in \cIB$. Recall that $\eps L^* \leq  \mu_{I} \leq 2L^*$. In the following equations, for notational simplicity, we omit the condition $X_{i,I}  = x'_{i,I},  i \in [h]$.  Recall that $\lambda = 1+ 3 \eps$.
    \begin{eqnarray*}
            && \Pr \Big[ \sum_{i \in [k]} X_{i, I} \geq \lambda \mu_I \Big] \qquad \\
        &=& \Pr  \Big[ \exp \Big( (\log \lambda) \sum_{i \in [k]} X_{i, I} \Big) \geq  \exp \Big( ( \lambda \log \lambda) \mu_I \Big)\Big]  \\
        &\leq&  \frac{ \Ex \Big[ \exp \Big( (\log \lambda) \sum_{i \in [k]} X_{i, I} \Big) \Big] }{\exp \Big( \lambda (\log \lambda) \mu_I\Big)} \qquad \mbox{ [By Markov's inequality]}\\
        &=&  \frac{ \prod_{i \in [k]} \Ex \Big[ \exp \Big( (\log \lambda)  X_{i, I} \Big) \Big] }{\exp \Big( \lambda (\log \lambda) \mu_I\Big)}  \qquad \mbox{ [Since $X_{i, I}$ are independent]}\\
        &=& E_{h, I}
    \end{eqnarray*}
     We now consider any $I \in \cIS$.
    \begin{eqnarray*}
        &&\Pr \Big[ \sum_{i \in [k]} X_{i, I} \geq  6 \eps L^* \Big] = \Pr  \Big[ \exp \Big(  \sum_{i \in [k]} X_{i, I} \Big) \geq  \exp \Big( 6 \eps L^* \Big)\Big]  \\
        &\leq&  \frac{ \Ex \Big[ \exp \Big(  \sum_{i \in [k]} X_{i, I} \Big) \Big] }{\exp ( 6 \eps L^* )} =  \frac{ \prod_{i \in [k]} \Ex \Big[ \exp (  X_{i, I} ) \Big] }{\exp ( 6 \eps L^* )} = E_{h, I}\\
    \end{eqnarray*}
    The last inequality is due to Markov's inequality.

\end{proofof}

\smallskip
\begin{proofof}[Lemma~\ref{lem:derand-1}]
    For each interval $I \in \cI$, it suffices to show that $E_{0,I} \leq \frac{1}{m^3}$. We first consider any $I \in \cIB$. Recall that $\lambda = 1 + 3\eps$ and $L^* \geq (1 /\eps^3) \log m $.
     \begin{eqnarray*}
        E_{0, I} &=& \frac{ \prod_{i \in [k]} \Big( \exp (\log \lambda) \mu_{ i, I} + (1 - \mu_{i, I}) \Big )}{ \lambda^{\lambda  \mu_I}} = \frac{ \prod_{i \in [k]} ( (\lambda -1)\mu_{ i,I} +1)}{ \lambda^{\lambda  \mu_I}}\\
        &\leq&  \frac{  \exp( 3\eps \sum_{i \in [k]} \mu_{i, I})    }{ (1 + 3\eps)^{ (1+3\eps) \mu_I} } =  \Big( \frac{  e^{ 3\eps}   }{ (1 + 3\eps)^{ 1+3\eps}} \Big)^{\mu_I}   \leq e^{-3 \eps^2 \mu_I}  = \frac{1}{m^3}
    \end{eqnarray*}
    The last inequality holds when $0 \leq \eps \leq 1/3$.

    Now consider any $I \in \cIS$.
    \begin{eqnarray*}
        E_{0, I} &=& \frac{ \prod_{i \in [k]} \Big( (e - 1) \mu_{i, I} + 1 \Big )}{ \exp( 6 \eps L^*)} \leq  \frac{  \exp( ( e - 1) \sum_{i \in [k]} \mu_{ i,I})    }{  \exp ( 6 \eps L^*)}  < \exp( - 3\eps L^*)   \leq \frac{1}{m^3}
    \end{eqnarray*}
\end{proofof}

\section{Throughput (Profit) Maximization}
    \label{sec:throughput}

\newcommand{\opts}{\opt^{small}}
\newcommand{\optl}{\opt^{large}}
\newcommand{\optlb}{\opt^{large}_{boundary}}
\newcommand{\optlm}{\opt^{large}_{middle}}
\newcommand{\optb}{\opt^B}
\newcommand{\len}{L}
\newcommand{\lenb}{B}
\newcommand{\lenh}{H}
\newcommand{\optlps}{\opt_{LP}^S}
\newcommand{\optlpl}{\opt_{LP}^L}
\newcommand{\efr}{\texttt{Fst}}
\newcommand{\Left}{\texttt{Left}}
\newcommand{\Right}{\texttt{Right}}
\newcommand{\esr}{\texttt{Snd}}
\newcommand{\ld}{\lambda}
\newcommand{\proxy}{\textrm{Proxy}}
\newcommand{\proxysel}{\textrm{ProxySelected}}
\newcommand{\Bad}{\texttt{Bad}}
\newcommand{\rs}{\cR^{small}}
\newcommand{\rl}{\cR^{large}}


In this section, we study the maximum throughput objective in broadcast scheduling. In this setting, each request $\rho$ is associated with its release time $r_\rho$, deadline $d_\rho$, weight (profit) $w_\rho$, and the page $p_\rho$ it asks for. We say request $\rho$ is satisfied within its window if page $p$ is transmitted during $[r_\rho+1, d_\rho]$. If we satisfy request $\rho$ within its window, we obtain the profit $w_\rho$ associated with the request $\rho$. The goal is to find a schedule that maximizes the total profit.

Our main result is an improved 0.7759-approximation for the maximum throughput objective. More precisely, we give a randomized algorithm with approximation factor $(1 / 2 + 3 / (4e) - \eps)$ for any $\eps >0$. Furthermore, the run time is $(1 / \eps)^{O(1/ \eps)}  \cdot poly(m)$. Note that for any fixed input size, the run time required to achieve an approximation factor arbitrarily close to $(1 / 2 + 3 / (4e))$ increases by a multiplicative factor depending only on $\eps$.

The analysis of the 0.7759-approximation is algebraically involved. Hence we first present a $0.754$-approximation to illustrate our main ideas behind the improvement. The 0.7759-approximation  will be presented in Section~\ref{sec:further}. Also we will first present an algorithm with run time $m^{(1 / \eps)^{O(1/\eps)}}$, and reduce the run time to $(1 / \eps)^{O(1/ \eps)}  \cdot poly(m)$ in Section~\ref{sec:throughput-runtime}.

\subsection{Overview of $0.754$-approximation}

Our algorithm improves upon the previous best known approximation by handling requests of large/small windows separately. Here we say that a request $\rho$ has a large window if $\rho$'s window length $|W_\rho| \geq 2 \lenh$, otherwise a small window; later $\lenh$ will be set as $1 / \eps^3$. Let $\rs$ and $\rl$ denote the set of small-window and large-window requests, respectively.

To introduce our integer/linear programming, we need to define a fair amount of notation. To collect most profits from small-window requests, we create configuration variables $y_{I, Q}$ for each sub-interval $I$ which is defined as follows. Let $I_0, I_1, ..., I_{h}$ be the minimum number of disjoint intervals that cover the time horizon $[0, T = \max_{\rho} d_\rho]$ seamlessly; the run time will depend on  $T$, but we will remove this dependency later. All intervals have a length $2\lenh/ \eps$ except the first interval.  We set the length of the first interval as a random number drawn from $[1, 2\lenh / \eps]$. In fact, by trying all possible values in $[1,2 \lenh / \eps]$, we can assume that we know the right value; what ``right" means will become clear soon. Note that the first interval determines all other intervals. We let $\cI := \{I_0, I_1, ..., I_h\}$. For each interval $I_i$, we define configurations. We say that $Q = \{ (t, q_t) \; | \; t \in I_i \}$ is a configuration with respect to $I_i$, where for all $ t \in I_i$, $q_t \in \cP$. Of course, the pair $(t, q_t)$ implies that  page $q_t$ is transmitted at time $t$. Let $\cQ(I_i)$ denote the collection of all configurations with respect to $I_i$. For each configuration $Q \in \cQ(I_i)$ we create a variable $y_{I_i, Q}$, which is 1 if the schedule follows $Q$ during $I_i$. Let $w_{I,Q}$ denote the total profit of the small-window requests arriving in $I$ that are satisfied (within their window) by the transmissions made by $Q$. Also for each request $\rho$ in $\rl$, we create a variable $z_\rho$ that indicates if $\rho$ is satisfied within its window or not. We note that the constant $\lenh$ will depend only on $\eps$, therefore the number of variables is at most a polynomial in  $m$ and $T$ for any fixed $\eps >0$ (In Section~\ref{sec:throughput-runtime}, we will show how to remove the dependency on $T$ and reduce the run time to $(1 / \eps)^{O(1/ \eps)}  \cdot poly(m)$).

We are now ready to set up our integer programming. Let $P(Q, I')$ denote the set of pages that are transmitted during $I'$ according to configuration $Q$.
\begin{alignat}{10} \tag{$\mathsf{IP}_\mathsf{Throughput}$} \label{lp:throughput}\\
    \max \qquad  \sum_{I \in \cI} \sum_{Q \in \cQ(I)} & w_{I, Q} y_{I,Q}+ \sum_{\rho \in \rl} w_\rho z_\rho    \\
     \sum_{Q \in \cQ(I)} y_{I,Q} &\leq 1 \qquad & \forall I \in \cI  \label{LP-throughput-3}\\
    \sum_{I \in \cI} \sum_{Q \in \cQ(I), p_\rho \in P(Q, I\cap W_\rho)}  y_{I, Q} &\geq z_\rho &\forall \rho \in \rl  \label{LP-throughput-5}\\
%
       y_{I,Q}  &\geq 0 &\forall I \in \cI, Q \in \cQ(I) \nonumber \\
       z_{\rho}  &\in \{0, 1\}&\forall \rho \in \rl \nonumber
\end{alignat}

Note that we use different variables $y_{I, Q}$ and $z_{\rho}$ to count profits from requests in $\rs$ and $\rl$, respectively. Further, it should be noted that not all the small-window requests contribute to the objective. More precisely, for each interval $I = [t_1, t_2] \in \cI$, the small-window requests arriving during $[t_2 - 2\lenh+1, t_2]$ may not contribute to $ \sum_{Q \in \cQ(I)} w_{I, Q} y_{I,Q}$. This can happen if the optimal solution satisfies some of those requests during the next interval $I'$ starting at $t_2 +1$. Because of this, we may lose some profits for those requests. However, for the ``right" choice of the length of the first interval in $\cI$, we will lose only $\eps$ fraction of profits from small-window requests, since small-window requests are discarded from only $\eps$ fraction of time steps.

 Constraints (\ref{LP-throughput-3}) restrict that only one configuration can be selected for each interval $I \in \cI$.  
 Constraints $(\ref{LP-throughput-5})$ follow from the definition of $z_\rho$.  A large-window request $\rho$ is satisfied if $p_\rho$ is transmitted during $I \cap W_\rho$ for some interval $I \in \cI$.
 We can relax this integer program into a linear program $\mathsf{LP}_\mathsf{Throughput}$ by replacing $z_{\rho}  \in \{0, 1\}$ with $0 \leq z_\rho \leq 1$.  We let $ x_{p, t} = \sum_{Q: (t,p) \in Q, Q \in \cQ(I)} y_{I,Q}$ for all $I \in \cI, t \in I,  p \in \cP$, which will relate configuration-based variables $y$ to time-indexed variables $x$.

We can obviously solve $\mathsf{LP}_\mathsf{Throughput}$ in polynomial time for any fixed precision factor constant $\eps$. Since for each $I \in \cI$, there can be at most $m^{2H / \eps} = m^{O( 1/ \eps^4)}$ variables, the run time will be also $m^{O( 1/ \eps^4)}$. As mentioned before, we will show how to reduce the run time down to $(1 / \eps)^{O(1/ \eps)}  \cdot poly(m)$ in Section~\ref{sec:throughput-runtime}. We will let $\opt$ denote the optimal fractional solution of the above LP. We will simply try all possible sizes of the first interval in $\cI$, and will take the best size that maximizes the objective in $\mathsf{LP}_\mathsf{Throughput}$. It is easy to see that $\opt$ is at least $(1 - \eps)$ times as large as the integral optimum. Let $\opts$ and $\optl$ denote the amount that small-window and large-window requests contribute to $\opt$, respectively (the first and second sum in the objective, respectively).

\medskip

We give a randomized rounding that achieves an expected total profit of at least $( 1 -  \eps )^2 \cdot ( 1 - 2/ 3e) \opt$ for any $\eps >0$. We develop two different rounding schemes and combine them. We first give a high-level overview of the algorithm and how we get an improved approximation factor. When $\opts$ is relatively large, we use an independent rounding that selects one configuration $Q \in \cQ(I)$ with probability $y_{I,Q}$ for each $I \in \cI$. Intuitively, since configurations consider enough details for satisfying small-window requests, we will be able to collect most profits from small-window requests. At the same time, due to the nature of independent rounding, we will be able to collect $(1 - 1/e) \optl$ from large-window requests. In Section~\ref{sec:small-window}, we prove that we can achieve a total expected profit of at least
\begin{equation} \label{eqn:profit-small}
    \opts+ (1 - 1/e)\optl
\end{equation}

On the other hand, when $\optl$ is relatively large, we use a different rounding scheme by modifying the $\alpha$-point rounding. Unlike the independent rounding, the $\alpha$-point rounding does not immediately yield a feasible schedule since there could be some time slots where more than one page is transmitted. Hence we need a certain contention resolution scheme. Here the difficulty is that if we move some transmissions to near time slots, they may become completely useless. This is because in the throughput objective, we can obtain a profit $w_\rho$  only when we can satisfy $\rho$ within its window. This could be the case particularly for small-window requests. Hence in this setting, we do not collect as much profits from small-window requests as when using the independent rounding scheme. Nevertheless, we will keep only one page using the fair contention resolution scheme in \cite{FeigeV10}, and will obtain a $(1 - 1/e)$-approximation, which is the same approximation guarantee that the independent rounding gives. For large-window requests, we will be able to collect more profits than $(1 - 1/e) \optl$ since large-window requests are less fragile to moving transmissions. In Section~\ref{sec:large-window}, we  show that we can collect a total expected profit of at least
\begin{equation} \label{eqn:profit-large}
 ( 1- \eps)^2 ((1 - 1/e)\opts+ (1  - 1/2e)\optl)
 \end{equation}

 Assuming that $\eps$ is arbitrarily small, the minimum of the above two quantities (\ref{eqn:profit-small}) and (\ref{eqn:profit-large}) is achieved when $\optl = 2 \opts$. Hence by selecting the better of these two solutions, we achieve a profit of at least
\begin{equation} \label{eqn:profit-mixed}
 ( 1- \eps)^2 (1 - 2 / 3e) \; \opt
 \end{equation}

Since the LP relaxation loses at most $\eps$ fraction of total profit from small-window requests compared to the optimal solution, we derive Theorem~\ref{thm:max-throughput}.

\subsection{When Small-window Requests Give Large Profits }
    \label{sec:small-window}

As mentioned previously, for each $I \in \cI$, we pick one configuration $Q \in \cQ(I)$ with probability $y_{Q, I}$.
By (\ref{LP-throughput-3}) and from the objective, we know that the expected profit we obtain from small-window requests is at least $ \opts$. Now consider any large-window request $\rho$. Let $I'_1, ..., I'_k \in \cI$ denote the time intervals that $\rho$'s window spans over, i.e. $r_{\rho} +1 \in I'_1$ and  $d_\rho \in I'_k$. Let $\gamma_{i}:= \sum_{Q: Q \in \cQ(I'_i): p_\rho \in P(Q, I'_i \cap W_\rho)} y_{Q, I'_i}$ denote the fraction of configurations with respect to $I'_i$ that satisfy request $\rho$. Note that we have $z_\rho = \min \{\sum_{i \in [k]} \gamma_{i}, 1\} $ due to  constraints (\ref{LP-throughput-3}) and (\ref{LP-throughput-5}).

Hence the probability that request $\rho$ is satisfied is at least
$$1- \prod_{i = 1}^{k} ( 1- \gamma_i) \geq 1- \prod_{i = 1}^{k} \exp( - \gamma_i)  = 1 - \exp ( - \sum_{i = 1}^k \gamma_i) \geq 1- \exp( - z_\rho) \geq (1 - 1/e) z_\rho$$
This gives us $(1 - 1/e) \optl$ profit in expectation from large-window requests. This proves the profit claimed in (\ref{eqn:profit-small}).

\subsection{When Large-window Requests Give Large Profits}
\label{sec:large-window}

This section is devoted to describing and analyzing an algorithm that achieves a profit as large as (\ref{eqn:profit-large}). This will be useful when large-window requests give a relatively large total profit.

\subsubsection{Algorithm}
	\label{sec:the-algorithm}
Our algorithm uses only $x_{p,t}$ and consists of several steps.

\medskip
\noindent
\textbf{1. $\alpha$-point Rounding}: The first step using the standard $\alpha$-point rounding we obtain an infeasible tentative schedule. Formally, for each page $p$, choose $\alpha_p$ uniformly at random from $[0, 1]$. Then transmit page $p$ at all times $t$  such that $\sum_{t' = 1}^{t-1} x_{p,t'} < \alpha_p + k \leq \sum_{t' = 1}^{t} x_{p,t'}$ for some integer $k \geq 0$. Let $A_t$ denote the set of all pages transmitted at time $t$, and let $\stemp$ denote the resulting infeasible schedule; $|A_t| >1$ could occur for some time $t$.

\medskip
\noindent
\textbf{2. Fair Contention Resolution}: In this step, the goal is to keep at most one page $p_t$ from $A_t$ at each time $t$; if $A_t = \emptyset$, $p_t$ may not exist. We will say that $p_t$ is the first-round page transmitted at time $t$. We use the fair contention resolution from \cite{FeigeV10}, and choose a page $p$ from $A_t$ with probability
$$\frac{1}{\sum_{q \in \cP} x_{q,t}} \Big( \sum_{q \in A_t \setminus \{p\}} \frac{x_{q,t}}{ |A_t| -1} + \sum_{q \in \cP \setminus A_t} \frac{x_{q,t}}{ |A_t| } \Big)$$
By Lemma 1.4 in \cite{FeigeV10}, it follows that
$$\Pr [ p_t = p ] \geq  \frac{1 - \prod_{q \in \cP} (1 - x_{q,t}) }{\sum_{q \in \cP} x_{q,t} } \cdot x_{p,t}\geq (1 - 1/e) x_{p,t},$$
since $\sum_{q \in \cP} x_{q,t} \leq 1$.

\medskip
To get more profits from large-window requests, we will try to reschedule pages in $A_t \setminus \{p_t\}$ to near time slots so that those pages can still satisfy some large-window requests that they could before rescheduling; we abuse the notation $\{p_t\}$ by letting it denote $\emptyset$ if $p_t$ does not exist. To make sure those pages can find empty time slots nearby (with a probability close to 1), we empty a small fraction of time slots. These freed-up time slots, together with time slots $t$ such that $|A_t| = 0$, will be used as new places to accommodate overflowed transmissions.

\medskip
\noindent
\textbf{3. Freeing-up Times}: Pick a random value $h$ uniformly at random from $[1, \eps H]$. Let $\cB: = \{B_i, 0 \leq i \leq \lceil T /  H \rceil +1\}$ denote a set of disjoint intervals where $B_0: = [0, h)$, and $B_i:= [ h+ (i-1) \eps H, h+ i  \eps H)$ for all $i \geq 1$. Note that all intervals in $\cB$ have the same length except the first interval $B_0$. To distinguish the intervals in $\cB$ from the intervals in $\cI$ used in defining $\mathsf{LP}_\mathsf{Throughput}$, we will call the intervals in $\cB$ blocks. Note that the blocks in $\cB$ cover the interval $[0, T]$ seamlessly. We now empty the last $\eps^2 H$ time slots from each block in $\cB$ by discarding all pages in $A_t$ at those times. We will say that those time slots are ``freed-up."

\medskip
\noindent
\textbf{4. Relocating Overflowed Transmissions (pages)}: The goal of this step is to with a good probability relocate all pages in $A_t \setminus \{p_t\}$ at non-freed-up times $t$ to nearby empty time slots. By empty time slots, we mean the ``freed-up" times plus the empty times $t$ from the beginning, i.e. $t$ such that $|A_t|  = 0$ in $\stemp$. We first decide in which direction to move overflowed transmissions. The direction will be decided either as right or left, each with a   probability $1/2$. Let $\Right$ and $\Left$ denote the former and the latter event, respectively. All overflowed transmissions follow the same direction. Formally, for each block $B = [t_1, t_2 + \eps^2 H) \in \cB$ (if $t_2 <0$, do nothing), we do the following: For each time $t \in [t_1, t_2)$ and $p \in A_t$,
\begin{itemize}
    \item if $p = p_t$, the page $p$ stays at time $t$.
    \item if $p \neq p_t$ and the event $\Left$ occurs, the page $p$ can move to an empty time in $[t_1 - \eps^2 H, t)$.
    \item if $p \neq p_t$ and the event $\Right$ occurs, the page $p$ can move to an empty time in $[t+1, t_2 + \eps^2 H)$.
\end{itemize}
Since this is the final schedule, at most one page can be transmitted at a time. Note that the perfect relocation for a block $B$ corresponds to a matching where all the transmissions at non-freed-up times in $B$ are covered by time steps in the interval $[t_1 - \eps^2 H, t_2)$ or $[t_1, t_2 + \eps^2 H)$ (which is determined by the event $\Left$ or $\Right$). We try to find such a matching. If unsuccessful, we keep only the first-round pages, i.e. $p_t$, $t \in [t_1, t_2)$ and throw out all other overflowed pages, i.e. $A_t \setminus \{p_t\}, t \in [t_1, t_2)$. This will suffice to give our claimed approximation guarantee.

\subsubsection{Analysis}

We say that page $p$ originates from time $t$ if it is transmitted
at time $t$ (in $\stemp$) by the $\alpha$-point rounding. When a page
originates from $t$, it may stay at the time $t$ ($p_t = p$), or
may be relocated to near time slots.  Also it may be discarded when it fails to find a new place. Let $\efr(\rho, t)$ denote the
event that $\rho$ is satisfied by the page $p_\rho$ that originates from
time $t$, and stays at the time $t$. Let $\esr(\rho, t)$ denote the
event that $\rho$ is satisfied by the page $p_\rho$ that originates from
time $t$ but  is successfully relocated to other time slots.

Throughout this section, we will consider all large-window requests $\rl$. Additionally we will also consider all small-window requests whose window is completely contained in an interval in $\cI$. For notational simplicity, by $\rs$, we will refer to only such small-window requests. Recall that we lose only $\eps$ fraction of small-window requests.  For a request $\rho \in \rs$, $z_\rho$ is defined as $\min \{ \sum_{t \in W_\rho} x_{p_\rho, t} , 1 \}$. We note that for any request $\rho$, $\sum_{t \in W_\rho} x_{p_\rho, t} \geq \sum_{I \in \cI} z_{\rho, I}$; recall that $z_{\rho, I} := \sum_{Q \in \cQ(I), p_\rho \in P(Q, I)} y_{I, Q}$. Both sides may not be equal when a configuration $Q \in \cQ(I)$ transmits page $p_\rho$ several times in $W_\rho \cap I$. Then the configuration view thinks of it as satisfying $\rho$ once while the time indexed view thinks of it as satisfying $\rho$ multiples times. There is no difference between these two views in the integer programming, but they can differ in linear programing relaxation.
Our achieved profit for small-window requests will be compared against the time-indexed LP profit, namely, $\sum_\rho w_\rho \cdot \min \{\sum_{t \in W_\rho} x_{p_\rho, t} , 1 \}$. This is justified since it is always no smaller than $\opts = \sum_{I \in \cI} \sum_{Q \in \cQ(I)}  w_{I, Q} y_{I,Q}$.

We will show the following two lemmas. As discussed in the algorithm's description, the first lemma easily
follows from the property of the Fair Contention Resolution scheme
in step (2).

\begin{lemma}
    \label{lem:all-requests}
    For all requests $\rho$ and any time $t$, $\Pr [ \efr(\rho, t) \; | \; p_\rho \in A_t] \geq (1 - \eps) ( 1-
1/e)$.
\end{lemma}

\begin{lemma}
    \label{lem:large-window-requests}
    Suppose $H \geq 1/ \eps^3$ and $0 < \eps \leq 1/100$. Then for all large-window requests $\rho \in \rl$ and any time $t$, $\Pr [
  \esr(\rho, t) \; | \; \neg \efr(\rho, t) \mbox{ and } p_\rho \in A_t] \geq
(1 - \eps)/2$.
\end{lemma}

Throughout the analysis, we make a simplifying assumption that
$\sum_{t \in W_\rho} x_{p_\rho, t} \leq 1$. Otherwise, it will only
help our analysis. Since all the events that $p_\rho \in A_t, t \in
W_\rho$ are disjoint, these two lemmas will imply that
a large-window request $\rho$ is satisfied within
its window with a probability of at least
\begin{eqnarray}
    &&\sum_{t \in W_\rho} \Pr[  \efr(\rho, t) \mbox { or }  \esr(\rho, t) \; | \; p_\rho \in A_t]  \cdot \Pr [p_\rho \in
    A_t] \nonumber\\
    &=&\sum_{t \in W_\rho} \Big ( \Pr[  \efr(\rho, t) \; | \; p_\rho \in A_t ] + \Pr [ \esr(\rho, t) \mbox{ and } \neg \efr(\rho, t)  \; | \; p_\rho \in A_t] \Big ) \cdot \Pr [p_\rho \in
    A_t] \nonumber\\
     &=&\sum_{t \in W_\rho} \Big(\Pr[  \efr(\rho, t) \; | \; p_\rho \in A_t] +  \Pr [
  \esr(\rho, t) \; | \;  \neg \efr(\rho, t) \mbox{ and }  p_\rho \in A_t] ( 1- \Pr [ \efr(\rho, t) \; | \; p_\rho \in A_t ])\Big)  \nonumber\\
  && \qquad \cdot \Pr [p_\rho \in
  A_t] \nonumber\\
  &\geq& \sum_{t \in W_\rho} (1 - \eps)^2 (1 - 1/2e) x_{p_\rho, t} \nonumber\\
  &=& (1 - \eps)^2 (1 -
  1/2e) z_\rho \label{eqn:410}
\end{eqnarray}

Also we can easily show from Lemma~\ref{lem:all-requests} that any request $\rho$ is satisfied with a probability of at least  $\sum_{t \in W_\rho} \Pr [
  \efr(\rho, t)] \geq (1 - \eps)(1 - 1/e) z_\rho$. Hence we will be
able to derive the profit claimed in (\ref{eqn:profit-large}).

\medskip
It now remains to show Lemma~\ref{lem:large-window-requests}.

\begin{proofof}[Lemma~\ref{lem:large-window-requests}]
Consider any large-window request $\rho$ and time $t'$. Throughout, we assume that the time $t'$ is not
freed-up, which is the case with probability $(1 -
\eps)$. Also assume that $p_\rho \in A_{t'}$ but $p_{t'} \neq p_\rho$. The proof will be conditioned on these events.

We first consider the event $\Right$. Let $B = [t_1, t_2 +
\eps^2 H) \in \cB$ be the block that includes time $t'$.
We begin with showing the following claim.

\begin{claim}
If for all times $t \in [t_1, t_2)$, $\sum_{\tau \in
[t , t_2)} |A_\tau| \leq t_2 - t + \eps^2 H$, then we can relocate
each overflowed page $p \in A_t \setminus \{p_t\}$, $t \in [t_1,
t_2)$, to the right, to one of the empty time slots in $B$ (including $p_\rho$ originating from $t'$).
\end{claim}
\begin{proof}
We consider overflowed
pages at time $t$ from $t_2$ to $t_1$, and move each of such pages to the right most empty time slot in $B$.
By an induction on time $t$, we can show that all overflowed pages in $A_\tau \setminus \{p_\tau \}$, $t \leq \tau \leq t_2 -1$ are relocated to an empty time slot in $(t, t_2 + \eps^2 \lenh)$.
This is obviously true when $t = t_2$ ($t_2$ was freed-up, so no overflow at time $t_2$). To complete the induction, it suffices to show that there are enough time slots in $[t, t_2 + \eps^2 \lenh)$ to accommodate all transmissions made in $\stemp$ during $[t, t_2)$, and this is exactly what the condition of the claim implies.
\end{proof}

Let $\Bad_t$ denote the event that $\sum_{\tau \in [t , t_2)}
|A_\tau| > t_2 + \eps^2 H -  t$. We will show that $ \Pr [ \Bad_t]
\leq \eps / H$. Then by taking a union bound over all $t \in [t_1,
t_2)$, we can show that the desired relocation can be done with a
probability of at least $1 - \eps$. Bounding the probability of this
bad event is very similar to the proof of
Lemma~\ref{lem:fmax-delay}. Let $X_p$ denote a 0-1 random
variable that has value one with probability $\mu_p:= \sum_{t \in
[t, t_2)} x_{p,t} - \lfloor \sum_{t \in [t, t_2)} x_{p,t} \rfloor$.
Let $\mu := \sum_{p \in \cP} \mu_p$. Note that
    $$\Pr [ \Bad_t]  \leq \Pr [ \sum_{p \in \cP} X_p \geq \mu + \eps^2 H] $$
Also note that $\Var [ \sum_{p \in \cP} X_p ]  = \sum_{p \in \cP} \Var [
X_p ] \leq \sum_{p \in \cP} ( \mu_p - \mu_p^2) \leq \sum_{p \in \cP} \sum_{t \in [t, t_2)}  x_{p,t} \leq H$.
Hence by applying Theorem~\ref{thm:bernstein} with $V \leq H$, $b
\leq 1$, and $\Delta = \eps^2 H$, we have
$$\Pr [ \Bad_t] \leq \exp(- \frac{\Delta^2}{ 2V + 2b\Delta / 3}) \leq \exp( - \eps^2 H /3) \leq \eps / H,$$
when $H \geq 1/\eps^3$ and $\eps \leq 1/100$. Hence we know that in the event $\Right$, all overflowed pages in $B$ can be
safely relocated with a probability of at least $(1 - \eps)$.
One can easily show that this is the case also for the event $\Left$. The only exception is when $B$ is the first block $B_0$ in $\cB$, since there are no empty slots before time 0.

Suppose $B$ is not the first block. Since $\rho$ is a large-window request (of length at least $2\lenh$),
its window $W_\rho$ must cover at least one of $t_1 - \eps^2 H$ or $t_2 + \eps^2
H$. Hence the request $\rho$ is satisfied by the page $p_\rho$ originating from time $t'$ for at least one of the two events $\Left$ and $\Right$, which move the page to $[t_1 - \eps^2 H, t')$ and $(t', t_2 + \eps^2 H)$, respectively. This implies that  the event $\esr(\rho, t')$ occurs with a probability of at least $(1 - \eps)/2$.
Now suppose $B$ is the first block. In this case, any large-window request $\rho$ that
starts during the first block covers the entire second block. Hence
in the event $\Right$, if the relocation is successful, the request
will be satisfied.
\end{proofof}

\subsection{Further Improvement: $0.7759$-approximation}
	\label{sec:further}

In this section we further improve our approximation guarantee. The algorithm remains almost the same (except small changes in preprocessing), but the improvement comes from a more refined analysis. We will give a $0.7759$-approximation, more precisely a randomized $(1 /2  + 3/(4e) - \eps)$-approximation for any $\eps >0$.  The main idea is to collect more profits from large-window requests. In the previous analysis we counted profits from small-window/large-window requests separately. Recall that by using the $\alpha$-point rounding together with the contention resolution scheme, we were able to satisfy a large-window request $\rho$ in the second round with about a half probability if it is not satisfied in the first round (see Lemma~\ref{lem:large-window-requests}). The reason for the half probability loss can be summarized as follows: When the $\alpha$-point rounding tries to transmit page $p_\rho$ near to $\rho$'s boundary, the transmission may be relocated to a time step which is out of $\rho$'s window. Since we moved the transmission either to the left or right, each with a half probability, the half loss could occur. However, if $\rho$ is satisfied in the ``middle" of
its window, we will be able to avoid such a loss. Hence our analysis will consider two cases depending on whether $\rho$ is satisfied near to its boundary or not.

\smallskip To this end, we preprocess large-window requests as well. Recall that the time horizon was divided into intervals $\cI = \{I_0, I_1, ..., I_h\}$ where all intervals have length $2H / \eps$ except the first. For each large-window request $\rho$, we divide its window into three sub-windows, $W^l_\rho, W^m_\rho$ and $W^r_\rho$: Consider all intervals in $\cI$ that intersect $W_\rho$, and $W^l_\rho$ and $W^r_\rho$ are the intersections with the first and last of those intervals, respectively. The middle window $W^m_\rho := W_\rho \setminus ( W^l_\rho \cup W^r_\rho)$ is defined as the remaining interval of $W_\rho$ other than $W^l_\rho$ and $W^r_\rho$. (If $W_\rho$ intersects only one interval in $I_\rho$, then $W^l_\rho := w_\rho$, and there are no $W^m_\rho$, $W^r_\rho$. If $W_\rho$ intersects exactly two intervals, there is no middle window for $\rho$). Recall that we discarded all small-window requests that are not fully contained in an interval in  $\cI$. We also discard some large-window requests. We discard a large-window request $\rho$ if $W^l_\rho$ or $W^r_\rho$ has a length less than $2 \eps H$; every large-window request has a window of length at least $2H$. Note that we lose only $4 \eps H  /  (2H) \leq 2\eps$. For simplicity we assume that $\eps \leq 1/ 2$, and hence we lose only $\eps$ fraction of total profit from large-window requests in expectation.  We let $W^b_\rho := W^l_\rho \cup W^r_\rho$ to denote possibly two ``boundary" windows together.



For any $I \in \cI$, configuration $Q \in \cQ(I)$, and for any interval $I'$, we let $P(Q, I')$ denote the set of pages transmitted during $I'$ in the configuration $Q$. Let $z_{\rho, I} := \sum_{Q  \in \cQ(I), p_\rho \in P(Q, W_\rho \cap I)}  y_{I, Q}$ denote the total fraction of configurations in the LP that satisfies $\rho$ in $I$.
Let $z^m_\rho := \min \{ \sum_{I \in \cI, I \in W^m_\rho} z_{\rho, I}, 1 \}$, and let $z^b_\rho := \min \{ \sum_{I \in \cI, I \in W_\rho} z_{\rho, I}, 1 \} -z^m_\rho$. Notice that $z^m_\rho + z^b_\rho = z_\rho$.  We split $\optl$ into two quantities, $\optlb: = \sum_\rho w_\rho z^b_\rho$ and $\optlm: = \sum_\rho w_\rho z^m_\rho$. Note that $\optl = \optlb + \optlm$. By the independent rounding we will show that we can get a total expected profit of at least
\begin{equation} \label{eqn:profit-small-2}
    \opts+ \frac{2}{e} \optlb + \frac{3}{2e}\optlm
\end{equation}

To have a feel how we get these constants $2 / e$ and $3 /(2e)$, suppose that a large-window request is satisfied only in its boundary windows $W^l_\rho$ or $W^r_\rho$, by at most one unit in total. Let $b_1$ and $b_2$ denote the amount of configurations that transmit page $p_\rho$ in $W^l_\rho$ and $W^r_\rho$, respectively.  Then we know that $b_1 + b_2 = z_\rho \leq 1$. The probability that request $\rho$ is satisfied is then $1 - (1 - b_1) (1 - b_2) = b_1 + b_2 - b_1 b_2 \geq b_1 + b_2 - (b_1 + b_2)^2 / 4 \geq (3/4) (b_1 + b_2) = (3/4) z_\rho$. In another extreme case when $\rho$ is satisfied only during $\rho$'s middle window, we know that $\rho$ is satisfied with a probability of at least $(1 - 1/e) z_\rho$.  Ideally, it would be great if we could obtain $(3 / 4) \optlb+ (1 - 1/e)\optlm$. However, it is not the case when $\rho$ is satisfied both in $W^b_\rho$ and $W^m_\rho$. Nevertheless, by careful analysis we will be able to show (\ref{eqn:profit-small-2}) in Section~\ref{sec:small-2}.

Also using the variant of $\alpha$-point rounding with the contention resolution scheme we will show that we can achieve a total profit of at least
\begin{equation} \label{eqn:profit-large-2}
 ( 1- \eps)^2 \Big((1 - \frac{1}{e})\opts+ (1  - \frac{1}{2e})\optlb + \optlm \Big)
 \end{equation}

 As discussed above, if the $\alpha$-point rounding tries to satisfy  a large-window request $\rho$ in the middle of $\rho$'s window, then we can almost always satisfy $\rho$ within its window, either in the first round or second round.

 \smallskip
By selecting the better of these two outcomes, we can get a total expected profit of at least the average of the two, hence we get an approximation guarantee of
\begin{eqnarray*}
   &&( 1- \eps)^2 \cdot \Big(  (1 - \frac{1}{2e}) \opts + ( \frac{1}{2} + \frac{3}{4e}) \optl \Big) \\
   &\geq&  ( 1- \eps)^2 \cdot ( \frac{1}{2} + \frac{3}{4e}) \opt \simeq 0.7759 ( 1- \eps)^2 \opt
\end{eqnarray*}

\subsubsection{Approximation Guarantee of Independent Rounding:  Proof of (\ref{eqn:profit-small-2}) }
	\label{sec:small-2}

In this section, we show that the independent rounding of picking one configuration $Q \in \cQ(I)$ with probability $y_{Q, I}$ gives an expected total profit claimed in  (\ref{eqn:profit-small-2}). As before, we can show that we get an expected profit of $\opts$ from small-window requests. Hence we focus on proving that for each large-window request $\rho$ (which was not discarded in the beginning), we get an expected profit of at least
\begin{equation}
	\label{eqn:310}
(2 / e) z^b_\rho + ( 3 / (2e)) z^m_\rho
\end{equation}
 By summing this lower bound over all large-window requests $\rho$, we can prove the profit claimed in (\ref{eqn:profit-small-2}) due to linearity in expectation.

Recall that $P(Q, I')$ denotes the pages that are transmitted during interval $I'$ by configuration $Q$. Also recall $z_{\rho, I} := \sum_{Q  \in \cQ(I), p_\rho \in P(Q, W_\rho \cap I)}  y_{I, Q}$ denotes the total fraction of configurations in the LP that satisfies $\rho$ in $I$.
For each large window request $\rho$ we define the following three quantities (for notational convenience we omit $\rho$ from the notation):
\begin{align*}
	\eta_1 &:= z_{\rho, I}		&\mbox{ for the unique $I \in \cI$ s.t. $W^l_\rho \subseteq I$} \\
	\eta_2 &:= \sum_{I \in \cI: I \subseteq W^m_\rho} z_{\rho, I} &\\
	\eta_3 &:= z_{\rho, I}		&\mbox{ for the unique $I \in \cI$ s.t. $W^r_\rho \subseteq I$}
\end{align*}
Intuitively, these quantities $\eta_1, \eta_2$ and $\eta_3$ represent how much the LP satisfies a large window request $\rho$ in its left, middle, right windows, respectively. Observe that $\eta_1, \eta_3 \leq 1$, but not necessarily $\eta_2 \leq 1$. Note that $z^b_\rho = \min \{\eta_1 + \eta_2 + \eta_3, 1 \} - z^m_\rho$ and $z^m_\rho = \min \{\eta_2, 1 \}$. 
The probability that $\rho$ is satisfied in its window $W_\rho$ is at least
\begin{eqnarray*}
	\label{eqn:201}
&&	(1 - ( 1 - \eta_1)( 1 - \eta_3)) + ( 1- \eta_1) ( 1 - \eta_3) ( 1 - \Pi_{I \in \cI: I \subseteq W^m_\rho} (1 -  z_{\rho, I} ))\\
&\geq& (1 - ( 1 - \eta_1)( 1 - \eta_3)) + ( 1- \eta_1) ( 1 - \eta_3) ( 1 - \exp( - \eta_2)) \\
&\geq& (1 - ( 1 - \eta_1)( 1 - \eta_3)) + ( 1- \eta_1) ( 1 - \eta_3) ( 1 - \exp( - z^m_\rho ) )\\
&\geq& (\eta_1 + \eta_3 - \eta_1 \eta_3) \exp( - z^m_\rho ) + 1 - \exp( - z^m_\rho )
\end{eqnarray*}

We will show the following lemma which will show (\ref{eqn:310}), thereby completing the proof of (\ref{eqn:profit-small-2}). The proof is fairly algebraic.

\begin{lemma}
It holds that
$\displaystyle (\eta_1 + \eta_3 - \eta_1 \eta_3) \exp( - z^m_\rho ) + 1 - \exp( - z^m_\rho )
\geq \frac2e z^b_\rho + \frac{3}{2e} z^m_\rho$.
\end{lemma}
\begin{proof}
	For any fixed values of $\eta_1 + \eta_3$ and $\eta_2$ (which fix $z^b_\rho$ and $z^m_\rho$), observe the left-hand-side is minimized when $\eta_1 = \eta_3$. Hence it suffices to show that
\begin{eqnarray*}
&& (\eta_1 + \eta_3 - (\eta_1 + \eta_3)^2/4) \exp( - z^m_\rho ) + 1 - \exp( - z^m_\rho )
\geq \frac2e z^b_\rho + \frac{3}{2e} z^m_\rho
\end{eqnarray*}
Also since $ (\eta_1 + \eta_3 - (\eta_1 + \eta_3)^2/4)$ increases in $\eta_1 + \eta_3$ (when $0 \leq \eta_1 + \eta_3 \leq 2$), and $z^b \leq \eta_1 + \eta_3$, it suffices to show that
\begin{eqnarray}
	\label{eqn:251}
&& (z^b_\rho - (z^b_\rho)^2  /4 ) \exp( - z^m_\rho ) + 1 - \exp( - z^m_\rho )
\geq \frac2e z^b_\rho + \frac{3}{2e} z^m_\rho
\end{eqnarray}

For notional convenience, let $x: = z^b_\rho$ and $y:= z^mb_\rho$ (These variables have nothing to do with the variables in the LP. This override will be in effect only in the proof of this lemma). Note that $x + y \leq 1$ and $x, y \geq 0$. By rearranging terms in (\ref{eqn:251}) it remains to show
$$g( x, y) := (1 -\frac{x}{2})^2 e^{-y} + \frac{2}{e}x + \frac{3}{2e} y \leq 1$$
Observe that for any fixed $x$, the function $g(x,y)$ is maximized when
\begin{OneLiners}
	\item Case (i) $y= 0$; or
	\item Case (ii) $\frac{\partial g}{\partial y} = 0$; or	
	\item Case (iii) $y= 1-x$.
\end{OneLiners}

We continue our analysis by considering these cases as follows.

\medskip
\noindent \emph{Case (i):} First consider the case when $y = 0$. Then it is easy to see that $g(x,0) = (1 - x/2)^2 + (2 /e) x$ is maximized when $x =0$ for any $0 \leq x \leq 1$, and $g(0, 0) = 1$, hence $g(x, 0) \leq 1$.

\medskip
\noindent \emph{Case (ii):} Now consider the case when $\frac{\partial g}{\partial y} = 0$, which yields $(1 - x/ 2)^2 e^{-y} = 3 / (2e)$. Using this we have
$g(x, y) =  \frac{3}{2e} (1 + y) + \frac{2}{e}x$. Also we observe that  $y+1 \leq e^y \leq \frac{2e}{3} ( 1 - \frac{x}{2})^2$. Hence we derive that
$g(x, y) =  \frac{3}{2e} (1 + y) + \frac{2}{e}x \leq ( 1- \frac{x}{2})^2 + \frac{2}{e}x$. As in the previous case, this is at most $1$ for any $x \in [0,1]$.

\medskip
\noindent \emph{Case (iii):}  It now remains to consider the final case when $y = 1 - x$. By plugging this in $g(x,y)$ we derive that
$$g(x) := g(x, 1 - x) = (\frac{x}{2} -1)^2 \exp( x -1) + \frac{1}{2e} x + \frac{3}{2e}$$

By simple calculations we have
\begin{eqnarray*}
	\label{eqn:254}
	g'(x) &:=& \frac{d}{dx} g(x) = \frac{x}{2} ( \frac{x}{2} - 1) e^{x -1} + \frac{1}{2e}\\
	g''(x) &:=& \frac{d^2}{d^2x} g(x) = ( \frac{x^2}{4}   - \frac{1}{2} ) e^{x-1}
\end{eqnarray*}
Since $g''(x) < 0$ for any $x \in [0,1]$, $g'(0) >0$ and $g'(1) <0$, there exists a unique $x_0 \leq [0,1]$ such that function $g(x)$ increases during $[0, x_0)$ and decreases during $(x_0, 1]$. From the facts that $g(0.75) < 0.993994$, $g(0.755) > 0.993996$, and $g(0.76) \leq 0.993994$, we derive that $0.75 \leq x_0 \leq 0.76$.

Hence we conclude that
$$\max_{0 \leq x \leq 1} g(x) = g(x_0) < (\frac{0.75}{2} -1)^2  \cdot \exp( 0.76 -1) + \frac{1}{2e} \cdot 0.76 + \frac{3}{2e} < 0.99889 <1$$
This completes the proof of the lemma.
\end{proof}


\subsubsection{Approximation Guarantee of the algorithm in Section~\ref{sec:the-algorithm}:  Proof of (\ref{eqn:profit-large-2}) }

As before we assume that $\sum_{t \in W_\rho} x_{p_\rho, t} \leq 1$, since otherwise it can only improve our approximation guarantee. Now we take a closer look at the proof of Lemma~\ref{lem:large-window-requests}, and observe the following:
\begin{lemma}
    \label{lem:large-window-requests-2}
    Suppose $H \geq 1/ \eps^3$ and $0 < \eps \leq 1/100$. Then for all large-window requests $\rho \in \rl$ and any time $t$, we have
    \begin{eqnarray*}
    \Pr [
  \esr(\rho, t) \; | \; \neg \efr(\rho, t) \mbox{ and } p_\rho \in A_t  \mbox{ and } t \in W^m_\rho ] &\geq&
(1 - \eps)/2 \\
\Pr [
  \esr(\rho, t) \; | \; \neg \efr(\rho, t) \mbox{ and } p_\rho \in A_t  \mbox{ and } t \in W^b_\rho ] &\geq&
(1 - \eps)
  \end{eqnarray*}
\end{lemma}
\begin{proof}
	The proof is very similar to that of Lemma~\ref{lem:large-window-requests}, hence we give a sketch of the proof highlighting differences.
	The proof of the first probability remains the same. The reason why we have a better probability for the second case is the following. Recall that in the proof of Lemma~\ref{lem:large-window-requests}, 
we moved a overflowed transmission to the right or left by at most $\eps H$ time steps. Any large-window request $\rho$ we did not discard in the beginning (when solving the LP) intersects each of the two boundary windows $W^l_\rho$ and $W^r_\rho$ by at lest $2\eps H$ time steps. Also recall that we have shown that when we try to move the ``overflowed'' transmissions in an interval $I \in \cI$ to the right or left, we can find enough empty time slots for those transmissions with a probability of at least $(1 - \eps)$. Hence regardless of the direction in which the transmission $(p_\rho, t)$ is moved, it finds an empty time slot with a probability of at least $( 1- \eps)$ and still satisfies $\rho$ in its window conditioned on the transmission not being chosen in the first round.
\end{proof}

We are now ready to complete the proof of (\ref{eqn:profit-large-2}). Since all the events that $p_\rho \in A_t, t \in
W_\rho$ are disjoint, the two bounds in the lemma will imply that
a large-window request $\rho$ is satisfied within
its window with a probability of at least
\begin{eqnarray}
    &&\sum_{t \in W_\rho} \Pr[  \efr(\rho, t) \mbox { or }  \esr(\rho, t) \; | \; p_\rho \in A_t]  \cdot \Pr [p_\rho \in
    A_t]  \nonumber\\
    &\leq&\sum_{t \in W^b_\rho} \Pr[  \efr(\rho, t) \mbox { or }  \esr(\rho, t) \; | \; p_\rho \in A_t]  \cdot \Pr [p_\rho \in
    A_t]  \label{eqn:420}\\
&& + \sum_{t \in W^m_\rho} \Pr[  \efr(\rho, t) \mbox { or }  \esr(\rho, t) \; | \; p_\rho \in A_t]  \cdot \Pr [p_\rho \in
    A_t]  \label{eqn:421}
\end{eqnarray}

By taking similar steps we did in showing (\ref{eqn:410}), we can show that
$(\ref{eqn:420}) \geq ( 1- \eps)^2 ( 1- 1/2e) z^b_\rho$ and $(\ref{eqn:421}) \geq ( 1- \eps)^2 z^m_\rho$. Hence we obtain

\begin{lemma}
	Suppose $H \geq 1/ \eps^3$ and $0 < \eps \leq 1/100$. Then any large-window request $\rho \in \rl$ is satisfied (in the first or second round) with a probability of at least
	$( 1- \eps)^2 \Big( ( 1- 1/2e) z^b_\rho +  z^m_\rho$\Big).
\end{lemma}

We use the same bound shown in Lemma~\ref{lem:all-requests} for small-window requests, and by summing over all request $\rho$, we complete the proof of the profit claimed in  (\ref{eqn:profit-large-2}).

\subsection{Improving the Run Time}
\label{sec:throughput-runtime}

\newcommand{\barP}{\bar{P}}

In this section, we show how to improve the run time of our algorithm to $(1 / \eps)^{O(1 / \eps)} \cdot poly(m)$. We first observe that we can easily remove the dependency on the time horizon length $T$. This can be done by taking similar steps as we took in Section~\ref{sec:fmax-removeT}. Here we explain an easy way of doing this: If there is an interval of length greater than $m$ on which no request's window starts or ends, we throw out all times steps in the interval but $m$ time steps. Note that one can obtain as much profit in in this ``down-sized" instance as in the original instance. Also it is easy to observe that the number of remaining time steps is at most $O(m^2)$.

The main idea for reducing the run time from $m^{O( 1/ \eps^4)}$ to $(1 / \eps)^{O(1 / \eps)} \cdot poly(m)$ is to solve the $\mathsf{LP}_\mathsf{Throughput}$ more efficiently. More concretely, we will solve the dual of the LP using an efficient separation oracle. The dual LP is defined as follows.
\begin{alignat}{10} \tag{$\mathsf{LP}_\mathsf{Dual:Throughput}$} \label{dual:throughput}\\
    \min \qquad  \sum_{I \in \cI} \gamma_I  + \sum_{\rho \in \rl} \xi_\rho  \nonumber \\
s.t. \;\;\;\;\;    w_{I, Q}  + \sum_{\rho \in \rl :p_\rho \in P(Q, I \cap W_\rho)} \delta_\rho &\leq  \gamma_I  \qquad & \forall I \in \cI, Q \in \cQ(I)  \label{dual-1}\\
			 - \delta_\rho + \xi_\rho &\leq w_\rho			&\forall \rho \in \rl \label{dual-2}\\
%
	       \gamma_I &\geq 0 &\forall I \in \cI \nonumber \\
		\delta_\rho, \xi_\rho &\geq 0 &\forall \rho \in \rl \nonumber
\end{alignat}

To avoid considering all
constraints (\ref{dual-1}), we use an efficient separation oracle. Since there are only polynomially many constraints
(\ref{dual-2}) (with no dependence on $\eps)$, we focus on giving a separation oracle for constraints (\ref{dual-1}).  We observe that for each fixed $I \in \cI$, the problem of finding a (if any) violated constraint (\ref{dual-1}) is essentially equivalent to finding a schedule during $I$ that maximizes the throughput when each small-window request $\rho$ has a profit $w_\rho$ and each large-window request $\rho$ has a profit $
\delta_\rho$. Here we assume that $I$ are the only time steps that exist (each request $\rho$'s window is restricted to $I$). Obviously we will need to transmit at most $2H/ \eps$ pages (the interval $I$ has a length of at most $2H/ \eps$). Hence at each time step $t$, as a potential page to transmit at the time, we only need to consider the $2H / \eps$ pages that yield the maximum profit assuming that the time step $t$ is the only time step when we transmit a page. Hence we can find the best schedule that maximizes the total profit in $I$ in time $(H/ 2 \eps)^{H / 2\eps} \cdot poly(m) = (1 / \eps)^{O(1 / \eps)} \cdot poly(m)$. If the configuration corresponding to this schedule violates the constraint
 (\ref{dual-1}), then we report it as a violated constraint. Otherwise, we conclude that no constraint is violated for the interval $I \in \cI$.

We have shown that one can solve the dual LP in $(1 / \eps)^{O(1 / \eps)} \cdot poly(m)$ time. This implies that we can find an optimal solution for  the dual LP by considering only $(1 / \eps)^{O(1 / \eps)} \cdot poly(m)$ many dual constraints. Hence by solving $\mathsf{LP}_\mathsf{Throughput}$ restricted only to the variables corresponding to those dual constraints, we can obtain an optimal solution for the primal $\mathsf{LP}_\mathsf{Throughput}$ by the strong duality theorem. The run time of our randomized algorithms is negligible.

\bibliographystyle{plain}

\clearpage
\appendix

\section{LP-guided $\fifo$}
\label{sec:LP-guided-fifo}
Since the algorithm $\fifo$ is $2$-competitive even in the online setting, one may hope that a natural modification of $\fifo$ may yield a better approximation. In this section, we show that a natural ``LP-guided" $\fifo$ does not improve the approximation ratio. The LP-guided $\fifo$ is defined as follows. Let $x^*_{p, t}$ be a fractional solution that satisfies all requests within $L^*$ time steps. At each time $t$ and for any page $p$, let $y^*_{p, t}: = \sum_{t' = t_1+1}^{t} x^*_{p, t'}$ where $t_1$ is the last time we transmitted page $p$; if no such time exists, then $t_1= 0$. Intuitively, $y^*_{p,t}$ suggests the amount or probability that page $p$ needs to be transmitted. Hence using $y_{p,t}$ to determine the priority, we transmit the page $p$ with the largest value of $y^*_{p,t}$ breaking ties favoring the page that has the earliest arriving unsatisfied request.

However, we can find a simple example that shows that this achieves only a 2-approximation. Consider the following instance. We will index pages by integers from 1 to $n$, and will assume that $n$ is even. At each time $t \in [1, n/2]$, two requests, each for page $2t - 1, 2t$, are released. The same sequence of requests is repeated during $[n/2 +1, n]$. That is, at each time $t + n/2$ for any $t \in [1, n/2]$, two requests, one for page $2t - 1$ and one for page $2t$ are released. This completes the description of the requests instance.

Now consider the following fractional solution $x^*_{p,t}$: For any page $p$, page $p$ is transmitted three times by half. More precisely, $x_{p, \lceil p /2 \rceil +1} = x_{p, \lceil p /2 \rceil + n / 2 +1} = x_{p, \lceil p /2 \rceil + n +1} = 1/2$, and for all other times $t$, $x_{p, t} = 0$. It is easy to check that the maximum flow time is $n/2 +1$ in this fractional solution. In contrast, the LP-guided $\fifo$ transmits all pages from $1$ to $n$ during $[2, n+1]$, and repeats this during $[n+2, 2n +1]$. In this schedule, the maximum flow time is $n$.

\newpage

\section{Pseudocodes}
    \label{sec:pseudo}

\begin{algorithm}[h!] \caption{Dynamic Programming with run time $O(T) \cdot n^{O(L^*)}$} \label{algo-constant} 
   \textbf{INPUT}: An estimate $L$ of the maximum flow time $L^*$; $\cP$ and $\cR$.
    \smallskip
    \begin{algorithmic}[1]
    \FOR {$t = L$ to $T$}
        \STATE $\cQ^f(t) \gets \emptyset$.
    \ENDFOR
   \STATE $\cQ^f(L-1) := \{ \{ ( 0, p_0) , (1, p_1), ..., (L-1, p_{L-1}) \} \; | \; p_0, p_1, ..., p_{L-1} \in \cP \}$.
    \FOR {$t = L$ to $T$}
        \FOR {each $Q  = \{ (t - L, p_{t - L}), (t - L+1 , p_{t - L+1}), ..., (t - 1, p_{t-1})\} \in \cQ^f(t - 1)$}
            \FOR {each $p_t \in \cP$}
                \IF {all requests with release time $t - L$ are satisfied by $p_{t - L+1}, p_{t - L+2}, ..., p_{t - 1}$ or $p_t$}
                          \STATE Let $Q' = \{(t - L+1, p_{t - L+1}), (t- L+2, p_{t - L+2}), ..., (t, p_t) \}$.
                    \STATE $\prev(Q') = Q$.
                    \STATE Add $Q'$ to $\cQ^f(t)$.
                \ENDIF
                \ENDFOR
        \ENDFOR
    \ENDFOR
    \smallskip
    \IF {$\cQ^f(T) = \emptyset$}
            \STATE Declare that the maximum flow time is greater than $L$.
    \ELSE
        \STATE Consider any  $Q \in \cQ^f(T)$.
        \FOR {$t = T$ to $L$}
        \STATE Let $(t, p_t) \in Q$.
        \STATE $p'_t \gets p_t$.
        \STATE $Q \gets \prev(Q)$.
        \ENDFOR
    \STATE Let $Q = \{(0, p_0), (1, p_1), ..., (L-1, p_{L-1})\}$.
    \STATE $p'_1 \leftarrow p_1$, $p'_2 \leftarrow p_2$, ..., $p'_L \leftarrow p_{L}$.
        \STATE Transmit page $p'_t$ at time $t$, $1 \leq t \leq T$.
        \ENDIF
    \end{algorithmic}
\end{algorithm}

\newpage

\begin{algorithm}[h!] \caption{Dynamic Programming when $L^* \leq (1 / \eps^2) \log T$}  \label{algo-small-2} 
   \textbf{INPUT}: Modified instance $J'$ with $\ell = 2 + 1/\eps$.  
    \smallskip
    \begin{algorithmic}[1]
    \STATE $T' := \lceil \frac{\max_\rho r_{\rho}}{ \eps L^*} \rceil + \ell$.
    \FOR {$t = \ell$ to $ T'$}
        \STATE $\cQ^f(t) \gets \emptyset$.
    \ENDFOR
   \STATE $\cQ^f(\ell-1) := \{ \{ (0, \cA_0), (1, \cA_1), ..., (\ell-1, \cA_{\ell-1}) \} \; | \;   \cA_i \in {\bigcup_{k=0}^{\ell-2} \cP_{k} \choose \eps L^*}, 0 \leq i \leq \ell-1 \}$.
    \FOR {$t = \ell$ to $T'$}   \label{algo-small-line-0}
        \FOR {each $Q  = \{ (t - \ell, \cA_{t - \ell}), (t - \ell+1 , \cA_{t - \ell+1}), ..., (t - 1, \cA_{t-1})\} \in \cQ^f(t - 1)$}    \label{algo-small-line-1}
            \FOR {each $\cA_t \in {\bigcup_{i=1}^\ell \cP_{t-i} \choose \eps L^*}$}                 \label{algo-small-line-2}
                \IF { $\cP_{t - \ell} \subseteq \bigcup_{i=0}^{\ell-1} \cA_{t - i}$}                     \label{algo-small-line-3}
                    \IF { no page appears more than twice in $\cA_{t - \ell+1}, \cA_{t - \ell +2}, ..., \cA_t$}                  \label{algo-small-line-3.5}
                          \STATE Let $Q' = \{(t - \ell+1, \cA_{t - \ell+1}), (t- \ell+2, \cA_{t - \ell+2}), ..., (t, \cA_t) \}$. \label{algo-small-line-4}
                    \STATE $\prev(Q') = Q$. \label{algo-small-line-5}
                    \STATE Add $Q'$ to $\cQ^f(t)$. \label{algo-small-line-6}
                    \ENDIF
                \ENDIF \label{algo-small-line-7}
                \ENDFOR
        \ENDFOR
    \ENDFOR \label{algo-small-line-10}
    \smallskip
    \IF {$\cQ^f(T') = \emptyset$}
            \STATE Declare that the maximum flow time is greater than $\ell$.
    \ELSE
        \STATE Consider any  $Q \in \cQ^f(T)$.
        \FOR {$t = T'$ to $\ell$}
        \STATE Let $(t, \cA_t) \in Q$.
        \STATE $\cA'_t \gets \cA_t$.
        \STATE $Q \gets \prev(Q)$.
        \ENDFOR
    \STATE Let $Q = \{(0, \cA_0), (1, \cA_1), ..., (\ell-1, \cA_{\ell-1})\}$.
    \STATE $\cA'_1 = \cA_1$, $\cA'_2 = \cA_2$, ..., $\cA'_{\ell-1} = \cA_{\ell-1}$.
        \STATE Transmit pages $\cA'_t$ at time $t$, $1 \leq t \leq T'$.
        \ENDIF
    \end{algorithmic}
\end{algorithm}

\newpage

\section{Concentration Inequalities}
    \label{sec:chernoff}

The following theorem follows from Bernstein inequalities.

\begin{theorem}[\cite{McDiarmid98concentration}]
    \label{thm:bernstein}
Let $X_1, X_2, ..., X_n$ be $n$ independent random variables such that for all $i \in [n]$, $X_i \leq b$. Let $Y = \sum_{i = i}^{n} X_i$, $\mu := \Ex[Y]$, and $V := \Var[Y]$. Then it follow that
$$\Pr \Big[ Y - \mu \geq \Delta \Big] \leq \exp\Big( - \frac{\Delta^2}{2V(1+ (b \Delta / 3V))}\Big)$$
\end{theorem}

\end{document}